\documentclass[12pt,technote, onecolumn]{IEEEtran}
\usepackage{latexsym}
\usepackage{mathrsfs}
\usepackage{multirow}
\usepackage{amsfonts,amssymb,amsmath,amsthm,bm}
\usepackage{color}
\usepackage{cite}
\newtheorem{theorem}{Theorem}

\newtheorem{definition}{Definition}
\newtheorem{remark}{Remark}

\newtheorem{lemma}{Lemma}

\usepackage{amsfonts}
\usepackage{mathrsfs}
\usepackage{amsmath}
\usepackage{supertabular}
\usepackage{multirow}
\usepackage{makecell}
\usepackage{longtable}
\usepackage{graphicx}
\usepackage{verbatim}
\usepackage{cite}
\usepackage{booktabs}
\usepackage{CJK}

\begin{document}

\title{Post-quantum Multi-stage Secret Sharing Schemes using Inhomogeneous Linear Recursion and Ajtai's Function}
\author{Jing Yang *, Fang-Wei Fu,
\IEEEcompsocitemizethanks{\IEEEcompsocthanksitem Jing Yang and Fang-Wei Fu are with Chern Institute of Mathematics and LPMC, and Tianjin Key Laboratory of Network and Data Security Technology, Nankai University, Tianjin 300071, P.R.China, Emails: yangjing0804@mail.nankai.edu.cn, fwfu@nankai.edu.cn.
}

\thanks{$^\dag$This research is supported by the National Key Research and Development Program of China (Grant No. 2018YFA0704703), the National Natural Science Foundation of China (Grant No. 12141108 and 61971243), and the Natural Science Foundation of Tianjin (20JCZDJC00610), the Fundamental Research Funds for the Central Universities of China (Nankai University), and Nankai University Zhou Haibing Zhide Foundation.}
\thanks{*Corresponding author}
}

\maketitle

\begin{abstract}
Secret sharing was firstly proposed in 1979 by Shamir and Blakley respectively. To avoid deficiencies of original schemes, researchers presented improvement schemes, among which the multi-secret sharing scheme (MSS) is significant. There are three categories of MSSs, however, we focus on multi-stage secret sharing scheme (MSSS) recovering secrets with any order in this work. By observing inhomogeneous linear recursions ($ILR$s) in the literature, we conclude a general formula and divide $ILR$s into two types according to different variables in them. Utilizing these two kinds of $ILR$s, we propose four verifiable MSSSs with Ajtai's function, which is a lattice-based function. Our schemes have the following advantages. Firstly, our schemes can detect cheat of the dealer and participants, and are multi-use. Secondly, we have several ways to restore secrets. Thirdly, we can turn our schemes into other types of MSSs due to the universality of our method. Fourthly, since we utilize a lattice-based function to mask shares, our schemes can resist the attack from the quantum computer with computational security. Finally, although our schemes need more memory consumption than some known schemes, we need much less time consumption, which makes our schemes more suitable facing limited computing power.
\end{abstract}


\section{Introduction}
\subsection{Related works}
The secret sharing scheme (SSS) is vital in real life, which was firstly proposed by Shamir \cite{S79} and Blakley \cite{B79} in 1979, respectively. The original SSS can distribute a predefined secret $S$ by assigning diverse values to corresponding parties as their shares. Generally speaking, we assume that there are $n$ parties $\mathcal{P}=\{P_{1},P_{2},\cdots,P_{n}\}$ in an SSS. Notice that only authorized subsets of these $n$ parties can utilize their shares to restore the secret $S$ by running a specific algorithm. All the authorized subsets constitute the access structure of an SSS denoted by $\Gamma \subseteq 2^{\mathcal{P}}$, where $2^{\mathcal{P}}$ is the power set of $\mathcal{P}$. It is desirable that the subset which is not in $\Gamma$ can get no information about the secret $S$. A special case of the general access structure is the $(t,n)$ threshold access structure with $t\leq n$, which consists of any at least $t$ parties of $\mathcal{P}$.

Since there are some disadvantages in the Shamir's SSS, many scholars proposed a lot of novel SSSs with new properties to improve the security level. For instance, some schemes can prevent cheat from the dealer and the parties simultaneously, verify the authenticity of the shares and reconstructed secret, and alter the threshold or the number of parties according to the security requirement. Nevertheless, these schemes above can merely share one secret. Then, the multi-secret sharing scheme (MSS) was proposed to distribute multiple secrets. However, in an MSS, it usually requires that each participant only keeps one share to recover all these secrets, and the size of each share is equal to every secret. Besides, most of the MSSs are proved to have computational security instead of information-theoretic security, which are less secure than the Shamir's scheme.

In general, there are three categories of MSS according to their different processes of secret reconstruction, including the simultaneous multi-secret sharing scheme (SMSS) \cite{YCH04,SC05,ZZZ07,DM081,DM082,DM083,HLC12,DM15,LZZ16,YF201,YF202}, the multi-stage secret sharing scheme recovering secrets with a predefined order (MSSSPO) \cite{HD94,H95,CHY05,M161,M162,M17}, and the multi-stage secret sharing scheme recovering secrets with any order (MSSSAO) \cite{FEA09,HRS13,FGEA14,HHC14,CLXW19}. What's more, depending on the specific practical situation, some category may be preferable. In an SMSS, all of the secrets are reconstructed simultaneously in only one stage. In the MSSSPO and MSSSAO, any authorized subset of the participants can recover only one secret in each stage. In the MSSSPO, different secret reconstruction phases must be executed in a predefined order, that is because the reconstructed secrets can leak some information about the unreconstructed secrets. However, in the MSSSAO, these secret reconstruction phases can be executed in any order, which means that the shared secrets can be recovered independently. Usually, MSSSPO and MSSSAO are called multi-stage secret sharing scheme (MSSS).

Most of the earlier proposed SMSSs \cite{YCH04,SC05,ZZZ07} are simple modifications of the Shamir's SSS \cite{S79}, where the dealer employs polynomials in order to distribute the secrets, and the authorized participants should use Lagrange Interpolation to recover these secrets. In 2008, for the first time, Dehkordi and Mashhadi \cite{DM081} employed the linear feedback shift register (LFSR) sequence instead of polynomials in SMSSs. The introduction of the LFSR sequence makes this SMSS have simpler way to recover the secrets and more efficient to adapt to the situation with limited computing power than presented schemes. Later, some researchers proposed a series of similar SMSSs \cite{DM082,DM083,HLC12,DM15} using the LFSR sequence due to its advantages. In 2016, Liu et al. \cite{LZZ16} pointed that the schemes in \cite{ZZZ07} and \cite{DM082} cannot resist the cheat from the dealer really and presented two improvement schemes to overcome this drawback. Furthermore, we find that the schemes in \cite{DM15} and \cite{DM083} also have the same disadvantage as mentioned above. Hence, based on \cite{LZZ16}, we have proposed two new SMSSs using LFSR public key cryptosystem \cite{YF201} and XTR public key system \cite{YF202}, respectively, which can detect the cheat from the dealer and participants simultaneously, then improve the efficiency further.

When it comes to the multi-stage secret sharing scheme (MSSS), the MSSSPO was proposed firstly by He and Dawson \cite{HD94} with one-way function and public shift technique in 1994. It should be noticed that they pointed that there is no possibility that any kind of MSSS can have unconditional security because of the information theoretic lower bound of \cite{KGH83}. Immediately, Harn \cite{H95} used the Lagrange interpolation polynomial to reduce the public values used in the improved scheme, where the memory consumption of the scheme is greatly reduced. Based on the scheme \cite{HD94}, Chang et al. \cite{CHY05} proposed an improvement scheme, a multi-use multi-stage secret sharing scheme, which reduces the time complexity of the scheme significantly. Here, multi-use means that the shares can be updated without the help of secure channel between the dealer and the corresponding parties, or can be reused directly since the shares are not leaked during the recovery phase. Later, in \cite{M161,M162,M17}, Mashhadi utilized the LFSR sequence to construct the schemes to recover secrets stage by stage, i.e., in a predefined order. Since these schemes increase efficiency further, the LFSR sequence is assumed to be an effective tool for constructing MSSSs.

As for the MSSSAO, in 2009, Fatemi et al. \cite{FEA09} proposed a scheme by all-or-nothing transform approach. This scheme achieves an effect that each secret owns its recovery function and the threshold of this scheme can change arbitrarily. Because of these properties, there is no restriction on the order of these secrets recovery. Then, in 2013, Herranz et al. \cite{HRS13} utilized symmetric encryption and the Shamir's scheme \cite{S79} to construct a multi-stage secret sharing scheme with computational provable security and they provided a formal security analysis for the first time in the literature. Their two schemes can be readily changed into the edition without the trusted third party, which can be applied to more real scenes. Also, in 2013, Fatemi et al. \cite{FGEA14} proved that any unauthorized parties can get some information about the shared secrets in SMSSs and recovered secrets can reduce the entropy of unreconstructed secrets in MSSSPOs by the knowledge of information theory. Therefore, by bilinear map, they presented a multi-use MSSSAO with less public values than before. After that, in 2014, Hsu et al. \cite{HHC14} proposed an ideal and perfect MSSSAO by using monotone span programs based on connectivity of graphs. However, this scheme does not have verifiability. In 2019, Chen et al. \cite{CLXW19} used information theory to construct a verifiable threshold multi-secret sharing scheme with different stages. They release different numbers of public values to control the threshold of every secret so that this scheme can recover every secret independently.

It can be seen from the previous description that all the MSSs mentioned so far now, no matter what category they belong to, they cannot resist the attack from the quantum computer. Because the knowledge used in these scheme is one-way hash functions, two-variable one-way functions, and other assumptions which can be solved by the quantum computer in polynomial time. In 1994, Shor \cite{S94} introduced the first quantum algorithm for factoring and computation of discrete logarithms, which made researchers pay more attention to post-quantum cryptography mainly consisting of code-based cryptography, lattice-based cryptography, multivariate public key cryptography and so on. Among them, because the lattice-based cryptography makes use of simple linear computation and has small size of the public key, it seems to play a more and more vital role than the other two methods in the post-quantum cryptography. Generally speaking, the security of lattice-based cryptography is built on the lattice problems in the worst case. Up to now, there is no algorithm proved to threat the security of lattice-based cryptography. Therefore, lattice-based cryptography is considered to be quantum resistant.

In a groundbreaking work \cite{A96}, for lattice problems, Ajtai gave the first lattice reduction from the worst-case to average-case, and therefore he proposed the first cryptographic object with a security proof supposing the intractability of well-studied computational problems of lattices. Specially, Ajtai provided a general way to construct a class of cryptographic one-way functions whose security is based on the worst-case hardness of lattice problems, which are still widely used in lattice-based cryptography. Afterwards, Goldreich et al. \cite{GGH96} showed that the functions in \cite{A96} are with collision resistance which is assumed to be more secure than one-wayness. Because lattice-based cryptography can resist the attack from the quantum computer, it is natural to use lattice-based cryptography devising the novel secret sharing schemes.

In 2012, Bansarkhani and Meziani \cite{BM12} utilized lattice-based one-way hash function to design an $(n,n)$ threshold secret sharing scheme with verifiability, which means that it can detect the deception by both the dealer and parties in this scheme. The security of this scheme is based on the hardness of $n^{c}$-approximate shortest vector problem (SVP), where $n$ is the dimension of the lattice used in this scheme and $c$ is a positive constant. Moreover, due to the usage of matrix vector operations in the verification phase of this scheme, compared with exponentiation operations in the traditional latticed-based secret sharing schemes, it reduces the computational complexity greatly. Then, in 2017, Pilaram and Eghlidos \cite{PE17} presented a $(t,n)$ threshold MSSSAO using Ajtai's function \cite{A96} . In this scheme, the shared secrets are vectors in a lattice with dimension $t$ and the dealer keeps the basis of the lattice secret. Since Ajtai's function is assumed to be quantum resistant, this MSSSAO provides computational security against the quantum computer at present. What's more, it has verifiability and is multi-use just by using simple linear lattice computation with high efficiency.

\subsection{Our results}
The motivation of our paper is to design an efficient verifiable multi-stage secret sharing scheme with any order. By using the inhomogeneous linear recursion ($ILR$) \cite{B02} and Ajtai's function \cite{A96}, we propose four verifiable MSSSAOs with computational security against the quantum computer.

From the analysis in the last subsection, we find that the LFSR sequence is an effective tool to construct the secret sharing scheme. In general, the LFSR sequence includes homogeneous and inhomogeneous linear recursions respectively. In this paper, we propose a general formula of $ILR$s used in the literature showing that there are two types of $ILR$s, i.e., Type-$t$ and Type-$l$ $ILR$s, which can be applied to SMSSs, MSSSPOs, and MSSSAOs. Besides, by using a standard model of MSSSAOs, we design four verifiable MSSSAOs with some special cases of Type-$t$ and Type-$l$ $ILR$s as examples to show the generality of our method.

The main contributions of our paper are as follows:

(1) We propose a general formula of some $ILR$s used in the literature;

(2) Our four schemes are verifiable, dynamic and reusable;

(3) We have more than one way to restore the shared secrets;

(4) Our method can be readily applied to SMSSs and MSSSPOs;

(5) Our MSSSAOs can resist the attack form quantum computation;

(6) Although our schemes need more memory consumption than some known schemes, we need much less time consumption.

The paper is organized as follows. In Section \uppercase\expandafter{\romannumeral 2}, we introduce the preliminaries used in this paper. In Sections \uppercase\expandafter{\romannumeral 3}, \uppercase\expandafter{\romannumeral 4}, \uppercase\expandafter{\romannumeral 5} and \uppercase\expandafter{\romannumeral 6}, we propose four new verifiable multi-stage secret sharing schemes recovering secrets with any order. Then we give security analysis in Section \uppercase\expandafter{\romannumeral 7} and performance analysis in Section \uppercase\expandafter{\romannumeral 8}. Finally, we conclude our paper in Section \uppercase\expandafter{\romannumeral 9}.

\section{Preliminaries}
\label{sec:1}
In this section, we introduce the knowledge of linear recursions, lattices and Ajtai's function, and the multi-stage secret sharing scheme.

\subsection{Inhomogeneous linear recursion}
In this subsection, we simply introduce the homogeneous and inhomogeneous linear recursion \cite{B02,LN97}.

Firstly, we give the definition of linear recurring sequence. Let $\mathbb{F}_{q}$ be a finite field with $q$ elements, where $q$ is a prime.

\begin{definition}
Assume that $k$ is a positive integer, and $c,a_{1},a_{2},\cdots,a_{k}$ are predefined constants over $\mathbb{F}_{q}$. Then, we call a sequence $\{u_{i}\}_{i\geq 0}$ over $\mathbb{F}_{q}$ satisfying the following relation
\begin{equation}
u_{i+k}+a_{1}u_{i+k-1}+a_{2}u_{i+k-2}+\cdots+a_{k}u_{i}=c \quad(i\geq 0)
\end{equation}
a $k$th-order linear recurring sequence over $\mathbb{F}_{q}$.
\end{definition}

When the relation (1) is determined, then $u_{0},u_{1},\cdots,u_{k-1}$ can define the whole sequence. Therefore, the values of the first $k$ terms are called initial values of this $k$th-order linear recurring sequence. When $c$=0, the relation (1) is considered to be homogeneous. When $c\neq0$, the relation (1) is considered as an inhomogeneous linear recursion ($ILR$) relation.

Now, let $\{u_{i}\}_{i\geq 0}$ satisfy a $k$th-order $ILR$ relation (1) over $\mathbb{F}_{q}$. By using (1) with $i$ replaced by $i+1$, we get
\begin{equation}
u_{i+k+1}+a_{1}u_{i+k}+a_{2}u_{i+k-1}+\cdots+a_{k}u_{i+1}=c \quad(i\geq 0).
\end{equation}
Then subtracting the relation (1) from the relation (2), we obtain
\begin{equation}
u_{i+k+1}+b_{1}u_{i+k}+b_{2}u_{i+k-1}+\cdots+b_{k}u_{i+1}+b_{k+1}u_{i}=0 \quad(i\geq 0),
\end{equation}
where $b_{1}=a_{1}-1$, $b_{j}=a_{j}-a_{j-1}$ for $j=2,\cdots,k$, and $b_{k+1}=-a_{k}$. So a $k$th-order inhomogeneous linear recurring sequence $\{u_{i}\}_{i\geq 0}$ can be interpreted as a $(k+1)$th-order homogeneous linear recurring sequence over $\mathbb{F}_{q}$. Therefore, the results about homogenous linear recurring sequences sometimes can be applied to the inhomogeneous case as well.

Then, we mainly talk about the consequences of the homogeneous case in the following part, and try to apply similar methods to the inhomogeneous case.

\begin{definition}
For a $k$th-order homogeneous linear recurring sequence $\{u_{i}\}_{i\geq 0}$ over $\mathbb{F}_{q}$,

(1) an expression of the following form
$$U(x)=\sum_{i=0}^{\infty}u_{i}x^{i}$$
is called its generating function over $\mathbb{F}_{q}$;

(2) the polynomial $$f(x)=x^{k}+a_{1}x^{k-1}+\cdots+a_{k} \in \mathbb{F}_{q}[x]$$
is called its characteristic polynomial.

\end{definition}


Next, we introduce a lemma \cite{LN97} before we give the main theorem.

\begin{lemma}
Suppose $\{u_{i}\}_{i\geq 0}$ is a $k$th-order homogeneous linear recurring sequence over $\mathbb{F}_{q}$, and $f(x)=(x-\alpha)^{k}$ is its characteristic polynomial where $\alpha\in \mathbb{F}_{q}$. Then for the generating function $U(x)=\sum_{i=0}^{\infty}u_{i}x^{i}$, the identity
$$U(x)=\dfrac{h(x)}{(1-\alpha x)^{k}},$$
holds, where $h(x)$ is a polynomial over $\mathbb{F}_{q}$ with $\deg (h(x))<k$.

Furthermore, $u_{i}=p(i)\alpha^{i}$ and $p(x)$ is a polynomial over $\mathbb{F}_{q}$ with degree at most $k-1$.

\end{lemma}

In fact, the method of generating functions in Lemma 1 can be applied to the $ILR$ satisfying the relation (1) sometimes. The applicability of the method corresponds to the specific form of $c$. Simply speaking, the method is feasible if the terms involving $c$ can be handled properly in a special way, which we will illustrate this process in the proof of Theorem 1 later.

Then we conclude a general formula of the $ILR$ relations used in the literature, which we will use in subsequent sections.

\begin{theorem}
We assume that a sequence $\{u_{i}\}_{i\geq 0}$ over $\mathbb{F}_{q}$ satisfies the $ILR$ relation as follows:

\begin{equation}
[ILR]=\left\{
\begin{aligned}
& u_{0}=c_{0},u_{1}=c_{1},\cdots,u_{t+l-2}=c_{t+l-2},\\
&\sum_{j=0}^{t+l-1}\left( {\begin{array}{*{20}{ccc}}
	t+l-1\\
	j
	\end{array}} \right)(-1)^{j}u_{i+t-(j-l+1)}=\frac{(i)_{l}}{l!}c \quad (i\geq 0),
\end{aligned}
\right.
\end{equation}
where $c,c_{0},c_{1},\cdots,c_{t+l-2}$ are given elements in $\mathbb{F}_{q}$, and $t,l$ are nonnegative integers. Here, we define that $(i)_{l}=i\times (i-1)\times \cdots \times (i-l+1)$. Therefore, $u_{i}=p(i)$, where $p(x)=A_{0}+A_{1}x+\cdots+A_{t+2l-1}x^{t+2l-1}$ is a polynomial over $\mathbb{F}_{q}$.

\end{theorem}

\begin{proof}
At first, we notice that $(i)_{l}=0$ when $i<l$. Then, we get the following equation:
\begin{footnotesize}
\begin{align*}
\bigg(\sum_{j=0}^{t+l-1}\left( {\begin{array}{*{20}{ccc}}
	t+l-1\\
	j
	\end{array}} \right)(-1)^{j}x^{j}\bigg)\sum_{i=0}^{\infty}u_{i}x^{i}
    &=u_{0}+(u_{1}-(t+l-1)u_{0})x+\cdots\\
    &+\bigg(\sum_{j=0}^{t+l-2}\left( {\begin{array}{*{20}{ccc}}
	t+l-1\\
	j
	\end{array}} \right)(-1)^{j}u_{t-j+l-2}\bigg)x^{t+l-2}\\
    &+\bigg(\sum_{j=0}^{t+l-1}\left( {\begin{array}{*{20}{ccc}}
	t+l-1\\
	j
	\end{array}} \right)(-1)^{j}u_{t-j+l-1}\bigg)x^{t+l-1}\\
    &+\bigg(\sum_{j=0}^{t+l-1}\left( {\begin{array}{*{20}{ccc}}
	t+l-1\\
	j
	\end{array}} \right)(-1)^{j}u_{t-j+l}\bigg)x^{t+l}+\cdots\\
    &\overset{(4)}{=}h_{t,l}(x)+cx^{t+2l-1}\bigg(1+\frac{(l+1)_{l}}{l!}x+\frac{(l+2)_{l}}{l!}x^{2}+\cdots\bigg)\\
    &=h_{t,l}(x)+cx^{t+2l-1}\cdot\frac{1}{(1-x)^{l+1}}\\
    &=\dfrac{h_{t,l}(x)(1-x)^{l+1}+cx^{t+2l-1}}{(1-x)^{l+1}},
\end{align*}
\end{footnotesize}

\noindent where
\begin{footnotesize}
$h_{t,l}(x)=u_{0}+(u_{1}-(t+l-1)u_{0})x+\cdots+\bigg(\sum_{j=0}^{t+l-2}\left( {\begin{array}{*{20}{ccc}}
	t+l-1\\
	j
	\end{array}} \right)(-1)^{j}u_{t-j+l-2}\bigg)x^{t+l-2}$
\end{footnotesize}
is a polynomial of $x$ with degree at most $(t+l-2)$, and $t,l$ are nonnegative integers. Hence,
$$(1-x)^{t+l-1}\sum_{i=0}^{\infty}u_{i}x^{i}=\dfrac{h_{t,l}(x)(1-x)^{l+1}+cx^{t+2l-1}}{(1-x)^{l+1}},$$
$$\sum_{i=0}^{\infty}u_{i}x^{i}=\dfrac{h_{t,l}(x)(1-x)^{l+1}+cx^{t+2l-1}}{(1-x)^{t+2l}}.$$
Then, according to Lemma 1, we obtain that $u_{i}=p(i)$, where $p(x)$ is a polynomial over $\mathbb{F}_{q}$ with degree at most $t+2l-1$. That is to say, $p(x)=A_{0}+A_{1}x+\cdots+A_{t+2l-1}x^{t+2l-1}$ is a polynomial over $\mathbb{F}_{q}$.

\end{proof}

\begin{theorem}
We assume that a sequence $\{u_{i}\}_{i\geq 0}$ over $\mathbb{F}_{q}$ satisfies the $ILR'$ relation as follows:

\begin{equation}
[ILR']=\left\{
\begin{aligned}
& u_{0}=c_{0},u_{1}=c_{1},\cdots,u_{t+l-2}=c_{t+l-2},\\
&\sum_{j=0}^{t+l-1}\left( {\begin{array}{*{20}{ccc}}
	t+l-1\\
	j
	\end{array}} \right)u_{i+t-(j-l+1)}=(-1)^{i}\frac{(i)_{l}}{l!}c \quad (i\geq 0),
\end{aligned}
\right.
\end{equation}
where $c,c_{0},c_{1},\cdots,c_{t+l-2}$ are given elements in $\mathbb{F}_{q}$, and $t,l$ are nonnegative integers. Here, we define that $(i)_{l}=i\times (i-1)\times \cdots \times (i-l+1)$. Therefore, $u_{i}=(-1)^{i}p(i)$, where $p(x)=A_{0}+A_{1}x+\cdots+A_{t+2l-1}x^{t+2l-1}$ is a polynomial over $\mathbb{F}_{q}$.

\end{theorem}

The proof of Theorem 2 is similar to the Theorem 1, so we omit this part here.\\


In fact, we find that the $ILR$ relations in Theorem 1 and Theorem 2 are general forms of the formulas used in some proposed schemes.

(1) Next, we talk about some special forms of (4) and (5) where $t$ is variable.

1) Firstly, as for the relation (4), when $l=0$, the $ILR$ relation over $\mathbb{F}_{q}$ is as follows which is used in \cite{M162} and \cite{DM15}:
$$
[ILR]_{l=0}=\left\{
\begin{aligned}
& u_{0}=c_{0},u_{1}=c_{1},\cdots,u_{t-2}=c_{t-2},\\
&\sum_{j=0}^{t-1}\left( {\begin{array}{*{20}{ccc}}
	t-1\\
	j
	\end{array}} \right)(-1)^{j}u_{i+t-j-1}=c \quad (i\geq 0),
\end{aligned}
\right.
$$
where $c,c_{0},c_{1},\cdots, c_{t-2}$ are given elements in $\mathbb{F}_{q}$. Therefore, $u_{i}=p(i)$, where $p(x)=A_{0}+A_{1}x+\cdots+A_{t-1}x^{t-1}$ is a polynomial over $\mathbb{F}_{q}$.

Then, when $l=1$, the $ILR$ relation over $\mathbb{F}_{q}$ is as follows which is used in \cite{DM083}:
$$
[ILR]_{l=1}=\left\{
\begin{aligned}
& u_{0}=c_{0},u_{1}=c_{1},\cdots,u_{t-1}=c_{t-1},\\
&\sum_{j=0}^{t}\left( {\begin{array}{*{20}{ccc}}
	t\\
	j
	\end{array}} \right)(-1)^{j}u_{i+t-j}=ci \quad (i\geq 0),
\end{aligned}
\right.
$$
where $c,c_{0},c_{1},\cdots, c_{t-1}$ are given elements in $\mathbb{F}_{q}$. Therefore, $u_{i}=p(i)$, where $p(x)=A_{0}+A_{1}x+\cdots+A_{t+1}x^{t+1}$ is a polynomial over $\mathbb{F}_{q}$.

2) Secondly, as for the relation (5), when $l=0$, the $ILR'$ relation over $\mathbb{F}_{q}$ is as follows which is used in \cite{M17} and \cite{DM15}:
$$
[ILR']_{l=0}=\left\{
\begin{aligned}
& u_{0}=c_{0},u_{1}=c_{1},\cdots,u_{t-2}=c_{t-2},\\
&\sum_{j=0}^{t-1}\left( {\begin{array}{*{20}{ccc}}
	t-1\\
	j
	\end{array}} \right)u_{i+t-j-1}=(-1)^{i}c \quad (i\geq 0),
\end{aligned}
\right.
$$
where $c,c_{0},c_{1},\cdots,c_{t-2}$ are given elements in $\mathbb{F}_{q}$. Therefore, $u_{i}=(-1)^{i}p(i)$, where $p(x)=A_{0}+A_{1}x+\cdots+A_{t-1}x^{t-1}$ is a polynomial over $\mathbb{F}_{q}$.

Then, when $l=1$, the $ILR'$ relation over $\mathbb{F}_{q}$ is as follows which is used in \cite{DM083}:
$$
[ILR']_{l=1}=\left\{
\begin{aligned}
& u_{0}=c_{0},u_{1}=c_{1},\cdots,u_{t-1}=c_{t-1},\\
&\sum_{j=0}^{t}\left( {\begin{array}{*{20}{ccc}}
	t\\
	j
	\end{array}} \right)u_{i+t-j}=(-1)^{i}ci \quad (i\geq 0),
\end{aligned}
\right.
$$
where $c,c_{0},c_{1},\cdots, c_{t-1}$ are given elements in $\mathbb{F}_{q}$. Therefore, $u_{i}=(-1)^{i}p(i)$, where $p(x)=A_{0}+A_{1}x+\cdots+A_{t+1}x^{t+1}$ is a polynomial over $\mathbb{F}_{q}$.

\begin{remark}
From the formulas above, we find that when $l$ increases by 1, the degree of $p(i)$ used in the general term of these $ILR$ relations increases by 2. For instance, when $l=2$, the $[ILR]_{l=2}$ is as follows:
$$
[ILR]_{l=2}=\left\{
\begin{aligned}
& u_{0}=c_{0},u_{1}=c_{1},\cdots,u_{t}=c_{t},\\
&\sum_{j=0}^{t+1}\left( {\begin{array}{*{20}{ccc}}
	t+1\\
	j
	\end{array}} \right)(-1)^{j}u_{i+t+1-j}=\frac{i(i-1)}{2}c \quad (i\geq 0),
\end{aligned}
\right.
$$
where $c,c_{0},c_{1},\cdots, c_{t}$ are given elements in $\mathbb{F}_{q}$. Therefore, $u_{i}=p(i)$, where $p(x)=A_{0}+A_{1}x+\cdots+A_{t+3}x^{t+3}$ is a polynomial over $\mathbb{F}_{q}$.
\end{remark}

\begin{remark}
When $l=1$ and $c=0$, the $ILR$ relation (4) degenerates into a homogeneous linear recursion relation in the following form which is used in \cite{DM082,HLC12,LZZ16,M161}:
$$
[ILR]_{l=1,c=0}=\left\{
\begin{aligned}
& u_{0}=c_{0},u_{1}=c_{1},\cdots,u_{t-1}=c_{t-1},\\
&\sum_{j=0}^{t}\left( {\begin{array}{*{20}{ccc}}
	t\\
	j
	\end{array}} \right)(-1)^{j}u_{i+t-j}=0 \quad (i\geq 0),
\end{aligned}
\right.
$$
where $c_{0},c_{1},\cdots, c_{t-1}$ are given elements in $\mathbb{F}_{q}$.

Because an $ILR$ relation with degree $k$ can be interpreted as a homogeneous linear recursion with degree $k+1$ using the equation (3), the relation (4) can also be considered as the general formula of some homogeneous linear recursion relation from this perspective. Similar conclusion is also true for the relation (5).

\end{remark}

(2) In the previous section, notice that the relation used in presented schemes is only the special case when $t$ is variable. As far as we know, there is no work using $ILR$ relation when $l$ is variable. Then, we talk about this situation.

1) Firstly, as for the relation (4), when $t=0$, the $ILR$ relation over $\mathbb{F}_{q}$ is as follows:
$$
[ILR]_{t=0}=\left\{
\begin{aligned}
& u_{0}=c_{0},u_{1}=c_{1},\cdots,u_{l-2}=c_{l-2},\\
&\sum_{j=0}^{l-1}\left( {\begin{array}{*{20}{ccc}}
	l-1\\
	j
	\end{array}} \right)(-1)^{j}u_{i+l-j-1}=\frac{(i)_{l}}{l!}c \quad (i\geq 0),
\end{aligned}
\right.
$$
where $c,c_{0},c_{1},\cdots, c_{l-2}$ are given elements in $\mathbb{F}_{q}$. Therefore, $u_{i}=p(i)$, where $p(x)=A_{0}+A_{1}x+\cdots+A_{2l-1}x^{2l-1}$ is a polynomial over $\mathbb{F}_{q}$.

Then, when $t=1$, the $ILR$ relation over $\mathbb{F}_{q}$ is as follows:
$$
[ILR]_{t=1}=\left\{
\begin{aligned}
& u_{0}=c_{0},u_{1}=c_{1},\cdots,u_{l-1}=c_{l-1},\\
&\sum_{j=0}^{l}\left( {\begin{array}{*{20}{ccc}}
	l\\
	j
	\end{array}} \right)(-1)^{j}u_{i+l-j}=\frac{(i)_{l}}{l!}c \quad (i\geq 0),
\end{aligned}
\right.
$$
where $c,c_{0},c_{1},\cdots, c_{l-1}$ are given elements in $\mathbb{F}_{q}$. Therefore, $u_{i}=p(i)$, where $p(x)=A_{0}+A_{1}x+\cdots+A_{2l}x^{2l}$ is a polynomial over $\mathbb{F}_{q}$.

2) Secondly, as for the relation (5), when $t=0$, the $ILR'$ relation over $\mathbb{F}_{q}$ is as follows:
$$
[ILR']_{t=0}=\left\{
\begin{aligned}
& u_{0}=c_{0},u_{1}=c_{1},\cdots,u_{l-2}=c_{l-2},\\
&\sum_{j=0}^{l-1}\left( {\begin{array}{*{20}{ccc}}
	l-1\\
	j
	\end{array}} \right)u_{i+l-j-1}=(-1)^{i}\frac{(i)_{l}}{l!}c \quad (i\geq 0),
\end{aligned}
\right.
$$
where $c,c_{0},c_{1},\cdots,c_{l-2}$ are given elements in $\mathbb{F}_{q}$. Therefore, $u_{i}=(-1)^{i}p(i)$, where $p(x)=A_{0}+A_{1}x+\cdots+A_{2l-1}x^{2l-1}$ is a polynomial over $\mathbb{F}_{q}$.

Then, when $t=1$, the $ILR'$ relation over $\mathbb{F}_{q}$ is as follows:
$$
[ILR']_{t=1}=\left\{
\begin{aligned}
& u_{0}=c_{0},u_{1}=c_{1},\cdots,u_{l-1}=c_{l-1},\\
&\sum_{j=0}^{l}\left( {\begin{array}{*{20}{ccc}}
	l\\
	j
	\end{array}} \right)u_{i+l-j}=(-1)^{i}\frac{(i)_{l}}{l!}c \quad (i\geq 0),
\end{aligned}
\right.
$$
where $c,c_{0},c_{1},\cdots, c_{l-1}$ are given elements in $\mathbb{F}_{q}$. Therefore, $u_{i}=(-1)^{i}p(i)$, where $p(x)=A_{0}+A_{1}x+\cdots+A_{2l}x^{2l}$ is a polynomial over $\mathbb{F}_{q}$.

\begin{remark}
From the description above, we discover that no matter how the value of $t$ changes, the right hand side of these $ILR$ relations remains unchanged. If we utilize this type of $ILR$ relation to construct MSSS, this property will reduce the memory consumption to some extent, because we do not need to choose different values of $c$ to satisfy different requirements of security level.
\end{remark}

In order to make a distinction between these two types of formulas, we call the first type of formula Type-$t$ $ILR$ relation, where $t$ is variable. And we call the second type of formula Type-$l$ $ILR$ relation, where $l$ is variable. Depending on the discussion above, we can use our relations (4) and (5) to construct different SMSSs, MSSSPOs, and MSSSAOs for different security requirements, and what we need to do is only to make some small changes on $t$ or $l$.

\subsection{Lattices and Ajtai's function}

In this section, we introduce the knowledge of lattices and a lattice-based one-way function, namely, Ajtai's function. Let $\mathbb{R}$ represent the field of real numbers and $\mathbb{Z}$ represent the ring of integers.

\begin{definition}
Let $\mathbf{b}_{1},\mathbf{b}_{2},\cdots,\mathbf{b}_{n}\in \mathbb{R}^{m}$ be $n$ linearly independent vectors. The lattice $\mathcal{L}$ generated by $\mathbf{B}=\{\mathbf{b}_{1},\mathbf{b}_{2},\cdots,\mathbf{b}_{n}\}$ is defined as the vector set of all linear combinations of $\mathbf{b}_{1},\mathbf{b}_{2},\cdots,\mathbf{b}_{n}$, and all the coefficients used here are chosen from $\mathbb{Z}$, that is to say,
$$
\mathcal{L}=\mathcal{L}(\mathbf{B})=\{z_{1}\mathbf{b}_{1}+z_{2}\mathbf{b}_{2}+\cdots+z_{n}\mathbf{b}_{n}:z_{1},z_{2},\cdots,z_{n}\in \mathbb{Z}\}.
$$
\end{definition}
\noindent The set of linearly independent vectors $\mathbf{B}=\{\mathbf{b}_{1},\mathbf{b}_{2},\cdots,\mathbf{b}_{n}\}$ is called a basis of $\mathcal{L}$ and $n$ is called the dimension or the rank of $\mathcal{L}$.\\

Generally speaking, the security of the lattice-based one-way function relies on the hardness of lattice problems, among which the shortest vector problem (SVP) and the closest vector problem (CVP) are the most famous lattice intractability problems. Then we give the definition of these two problems, and notice that here $\parallel \cdot \parallel$ denotes an arbitrary norm.

\begin{definition}
Given a lattice $\mathcal{L}=\mathcal{L}(\mathbf{B})$, where $\mathbf{B}$ is an arbitrary basis, the shortest vector problem (SVP) is to find a non-zero vector $\mathbf{v}\in \mathcal{L}$ satisfying
$$||\mathbf{v}|| \leq ||\mathbf{u}||$$
for any non-zero vector $\mathbf{u}\in \mathcal{L}$.

Given a lattice $\mathcal{L}=\mathcal{L}(\mathbf{B})$, where $\mathbf{B}$ is an arbitrary basis, and a target vector $\mathbf{t}\in \mathcal{L}$, the closest vector problem (CVP) is to find a non-zero vector $\mathbf{v}\in \mathcal{L}$ satisfying
$$||\mathbf{v}-\mathbf{t}|| \leq ||\mathbf{u}-\mathbf{t}||$$
for any non-zero vector $\mathbf{u}\in \mathcal{L}$.
\end{definition}

In fact, especially significant to lattice-based cryptography are approximate versions of SVP and CVP mentioned above, which are described by an approximation factor $\gamma\geq 1$. Next, the $\gamma$-approximate SVP and CVP are defined as follows:

\begin{definition}

Given a lattice $\mathcal{L}=\mathcal{L}(\mathbf{B})$, where $\mathbf{B}$ is an arbitrary basis, the $\gamma$-approximate SVP is to find a non-zero vector $\mathbf{v}\in \mathcal{L}$ satisfying
$$||\mathbf{v}||\leq \gamma ||\mathbf{u}||$$
for any non-zero vector $\mathbf{u}\in \mathcal{L}$.

Given a lattice $\mathcal{L}=\mathcal{L}(\mathbf{B})$, where $\mathbf{B}$ is an arbitrary basis, and a target vector $\mathbf{t}\in \mathcal{L}$, the $\gamma$-approximate CVP is to find a non-zero vector $\mathbf{v}\in \mathcal{L}$ satisfying
$$||\mathbf{v}-\mathbf{t}||\leq \gamma ||\mathbf{u}-\mathbf{t}||$$
for any non-zero vector $\mathbf{u}\in \mathcal{L}$.
\end{definition}

Note that when $\gamma=1$, we get SVP and CVP in Definition 4.

\begin{definition}
If a lattice $\mathcal{L}$ satisfies $q\mathbb{Z}^{n}\subseteq \mathcal{L} \subseteq \mathbb{Z}^{n}$ for some (possibly prime) integer $q$, the lattice is called the $q$-ary lattice.
\end{definition}

For example, we can get a $q$-ary lattice with dimension $m$ from a set of vectors $\mathbf{x}\in \mathbb{Z}^{m}$ satisfying the equality $A\mathbf{x}=0 \mod q$, where $A$ is a random matrix from $\mathbb{Z}_{q}^{n\times m}$ and $q$ is the modulus. That is to say, the $q$-ary lattice is expressed by
$$\mathcal{L}_{q}^{\bot}(A)=\{\mathbf{x}\in \mathbb{Z}^{m}:A\mathbf{x}=0 \mod q\}$$
For the remainder of this paper, we restrict our attention to $q$-ary lattices only and we assume that $q$ is a prime.\\

Then, we introduce a lattice-based one-way function, called Ajtai's function \cite{A96}. Given a uniformly random matrix $A\in \mathbb{Z}_{q}^{n\times m}$ and $\mathbf{x}\in \{0,1\}^{m}$, Ajtai's function is $$f_{A}(\mathbf{x})=A\mathbf{x }\mod q.$$

Inverting this function causes the problem as follows:

\begin{definition}
Given $q,m,n\in \mathbb{N}$ satisfying $m> n \log q$, and $q=O(n^{c})$ for some positive constant $c$, a random matrix $A$ chosen uniformly from $\mathbb{Z}_{q}^{n\times m}$ and a vector $\mathbf{y}=A\mathbf{x} \mod q$ from $\mathbb{Z}_{q}^{n}$ for some random vector $\mathbf{x}\in \{0,1\}^{m}$, it is hard to get the vector $\mathbf{x}$.
\end{definition}

In \cite{A96}, Ajtai showed that it is impossible to invert this one-way function with non-negligible probability, which is because inverting this function means that solving any instance of $n^{c}$-approximate SVP.

\subsection{The multi-stage secret sharing scheme}
Secret sharing was designed to deal with the secret key distribution. Note that the $(t,n)$-threshold SSS is our main consideration. In a $(t,n)$-threshold SSS, we can assign a shared secret $S$ to $n$ parties $\mathcal{P}=\{P_{1},P_{2},\cdots,P_{n}\}$ with $n$ different values, i.e., shares $\{sh_{1},sh_{2},\cdots,$ $sh_{n}\}$, and at least $t$ components of $\mathcal{P}$ which we call qualified sets can restore $S$ by pulling their shares together. All qualified sets $\{A\subseteq \mathcal{P}: |A|\geq t\}$ ($t\leq n$) constitute the threshold access structure $\Gamma \subseteq 2^{\mathcal{P}}$.\\

Then we introduce two vital definitions in SSS.

\begin{definition}

An SSS is \textbf{perfect} if it satisfies two characteristics as follows:

(1) Any qualified sets of $\mathcal{P}$ can restore the secret by their own shares.

(2) For parties which are not in the access structure $\Gamma$, they can not get any information about the secret $S$.
\end{definition}

\begin{definition}
An SSS is \textbf{ideal} if every share has the same size as the secret.
\end{definition}

Next, we consider multi-secret sharing schemes (MSSs). In general, we assume that there are $k$ secrets $\{S_{1},S_{2},$ $\cdots,S_{k}\}$ to be distributed, and only one share $sh_{j}$ $(1\leq j\leq n)$ needs to be stored by every participant $P_{j}$, which means that every share needs to be reused to restore $k$ secrets. Obviously, their shares should be masked in order to obtain corresponding shadows so that they can be reused to recover multiple secrets securely. As we stated before, there are three categories of MSSs. In this paper, we only take MSSSAO as an example to show the generality of our method. Hence, to better illustrate our schemes, we provide a standard model of the verifiable MSSSAO.

\begin{definition}

A verifiable multi-stage secret sharing scheme recovering secrets with any order $\Omega=(Set,Con,Rec,Ver)$ consists of four phases in the following form:

\begin{itemize}
  \item In the \textbf{Setup phase}, owning $n$ participants $\mathcal{P}=\{P_{1},P_{2},\cdots,P_{n}\}$ and $k$ independent access structure $\Gamma_{1},\Gamma_{2},\cdots,\Gamma_{k}$, the dealer mainly generates original shares $sh_{j}$ for every participant $P_{j}$ for $1\leq j\leq n$, and some public messages.
  \item In the \textbf{Construction phase}, the dealer utilizes $n$ shares $\{sh_{1},sh_{2},\cdots,sh_{n}\}$ and public massages generated in the last phase, and $k$ secrets $\{S_{1},S_{2},$ $\cdots,S_{k}\}$ to obtain corresponding shadows, even subshadows and other public messages.
  \item In the \textbf{Recovery phase}, before recovering the shared secrets, these participants can use public messages to verify the authenticity of these shares. Then, any authorized set can pull their shadows to get corresponding subshadows and restore every secret $S_{i}$ for $ 1\leq i\leq k$ independently. Notice that the recovered secrets are not possibly correct.
  \item In the \textbf{Verification phase}, any participant in the authorized set can test the correctness of their recovered secrets $S_{i}$ for $ 1\leq i\leq k$ by public messages.

\end{itemize}

\end{definition}

Because in the MSSSAO, from the recovered secrets, the adversary cannot get any information about the unrecovered secrets, every secret $S_{i}$ $(i=1,2,\cdots,k)$ owns its unique access structure $\Gamma_{i}$ independently. For instance, in a threshold MSSSAO, the access structure of one secret $S_{i}$ has the following form $\Gamma_{i}=\{A \subseteq \mathcal{P}: |A|\geq t_{i}\}$ ($i=1,2,\cdots,k$).

Then, we propose the \textbf{security model} of the MSSSAO $\Omega=(Set,Con,Rec,$ $Ver)$ with several conditions:

\textbf{Correctness}

For $i\in \{1,2,\cdots,k\}$ and some subset $A\in \Gamma_{i}$, the MSSSAO can reconstruct the corresponding secret correctly by the ways mentioned in the Recovery phase.

\textbf{Verifiability}

The participants can verify the correctness or the validity of the recovered secrets.

\textbf{Privacy}

(1) The adversary can get no useful information about the unrecovered secrets from the recovered secrets.

(2) Any number of participants less than the threshold cannot restore the secrets.\\

Later, we will propose four verifiable MSSSAOs with two types of $ILR$s. In the first two schemes, we use two different Type-$t$ $ILR$s. For every secret, we need choose different constant vectors to define diverse access structures for diverse secrets. In the last two schemes, we use two different Type-$l$ $ILR$s. Under this circumstance, we only need select one constant vector to define diverse access structures for every secret.

\section{Scheme 1}

In this scheme, let $\mathcal{D}$ be the dealer, $\mathcal{P}=\{P_{1},P_{2},\cdots,P_{n}\}$ be the collection of participants, and $\mathcal{S}=\{S_{1},S_{2},\cdots,S_{k}\}$ be a set of secrets. $S_{i}$ is from $\mathbb{F}_{q}^{t_{i}}$ where $t_{i}$ ($0< t_{i}\leq n$) is the threshold corresponding to every $S_{i}$ $(1\leq i\leq k)$, and $q$ is a big prime number. It should be noted that all matrix operations are performed on $\mathbb{F}_{q}$. Especially, the $ILR$ relation used here is a special case of relation (4) which is a Type-$t$ $ILR$ relation, and we set $l=1$.

\subsection{Setup phase}
(1) $\mathcal{D}$ randomly selects $sh_{j}\in \{0,1\}^{r}$ $(j=1,2,\cdots,n)$ as $P_{j}$'s share such that $sh_{i}\neq sh_{j}$ for $i\neq j$ where $r\geq \max \{t\log t, \log n\}$ with $t=\max\{t_{1},t_{2},\cdots,t_{k}\}$, and distributes the share vector $sh_{j}$ to the participant $P_{j}$ via a secure channel.

(2) $\mathcal{D}$ chooses uniformly random matrices $G_{i}\in \mathbb{F}_{q}^{t_{i} \times r}$ for $i=1,2,\cdots,k$.

(3) $\mathcal{D}$ chooses a uniformly random matrix $F\in \mathbb{F}_{q}^{t\times r}$, and computes $h_{j}=F\cdot sh_{j}$ for $j=1,2,\cdots,n$ by using Ajtai's function.

(4) $\mathcal{D}$ computes $H(S_{i})$ for $i=1,2,\cdots,k$, where $H(\cdot)$ is any public one-way hash function.

(5) $\mathcal{D}$ publishes $G_{1},G_{2},\cdots,G_{k},F,h_{1},h_{2},\cdots,h_{n},H(S_{1}),H(S_{2}),\cdots,H(S_{k})$ on the bulletin board.

\subsection{Construction phase}
$\mathcal{D}$ generates the subshadow $u_{i,j}$ for each secret $S_{i}$ $(1\leq i\leq k)$ and each participant $P_{j}$ $(1\leq j\leq n)$, and two extra $u_{i,n+1}$, $u_{i,n+2}$ for each secret $S_{i}$ $(1\leq i\leq k)$ as follows:

(1) Randomly chooses $k$ different constant vectors $c_{i}\in \mathbb{F}_{q}^{t_{i}}$, $i=1,2,\cdots,k$.

(2) For $i=1,2,\cdots,k$, computes shadow $d_{j}^{i}=G_{i}\cdot sh_{j}$ for $j=1,2,\cdots, t_{i}-1$ by using Ajtai's function.

(3) For $i=1,2,\cdots,k$, considers an inhomogeneous linear recurring sequence $\{u_{i,j}\}_{j \geq 0}$ defined by the following formula:
$$[*]_{l=1}=\left\{
\begin{aligned}
& u_{i,0}=S_{i},u_{i,1}=d_{1}^{i},\cdots,u_{i,t_{i}-1}=d_{t_{i}-1}^{i},\\
& \sum_{\lambda=0}^{t_{i}}\left( {\begin{array}{*{20}{ccc}}
	t_{i}\\
	\lambda
	\end{array}} \right)(-1)^{\lambda}u_{i,j+t_{i}-\lambda}=c_{i}j\quad (j\geq 0),
\end{aligned}
\right.$$
and calculates $u_{i,j}$ for $1\leq i\leq k$ and $t_{i}\leq j\leq n+2$.

(4) Calculates $r_{i,j}=u_{i,j}-d_{j}^{i}$ for $1\leq i\leq k$ and $t_{i}\leq j\leq n$.

(5) Publishes $c_{1},c_{2},\cdots,c_{k}$, $r_{i,j}$ for $1\leq i\leq k$, $t_{i} \leq j\leq n$, $u_{i,n+1}$ and $u_{i,n+2}$ for $1\leq i\leq k$.

\subsection{Recovery phase}

When the participant $P_{j}$ receives his share $sh_{j}$, $1\leq j\leq n$, he examines whether the following equation holds:
$$F\cdot sh_{j}\overset{?}{=}h_{j},\qquad 1\leq j\leq n.$$
If the shares are proved to be authentic, any authorized participants can recover the secrets with any order by the following two ways according to different conditions they meet.\\

$\textbf{Way\: 1}$: For $1\leq i\leq k$, assume that any $t_{i}$ participants $\{P_{j}\}_{j\in I_{i}}$ $(I_{i}=\{i_{1},i_{2},\cdots,i_{t_{i}}\}$ $\subseteq \{1,2,\cdots,n\})$ reconstruct the secret $S_{i}$ together. At first, these participants collect $t_{i}$ shadows $\{d_{j}^{i}\}_{j\in I_
{i}}$ to calculate $t_{i}$ subshadows as follows:
$$u_{i,j}=\left\{
\begin{aligned}
& d_{j}^{i},\qquad\quad\:\, if\quad 1\leq j\leq t_{i}-1;\\
& d_{j}^{i}+r_{i,j}, \quad if\quad t_{i}\leq j\leq n.
\end{aligned}
\right.$$

Method 1. They utilize $t_{i}$ subshadows $\{u_{i,j}|j\in I_{i}\}$ and published $\{u_{i,n+1}$, $u_{i,n+2}\}$ to solve the following Vandermond equations for $1\leq i\leq k$ and $1\leq s\leq t_{i}$:

$$       
\left(                 
  \begin{array}{ccccc}   
    1 & i_{1} & i_{1}^{2} \cdots & i_{1}^{t_{i}+1}\\  
    1 & i_{2} & i_{2}^{2} \cdots & i_{2}^{t_{i}+1}\\
    \vdots & \vdots & \vdots & \vdots \\
    1 & i_{t_{i}} & i_{t_{i}}^{2} \cdots & i_{t_{i}}^{t_{i}+1}\\ 
    1 & n+1 & (n+1)^{2} \cdots & (n+1)^{t_{i}+1} \\ 
    1 & n+2 & (n+2)^{2} \cdots & (n+2)^{t_{i}+1}  
  \end{array}
\right)\quad
\left(                 
  \begin{array}{c}   
    A_{i,0}^{[s]}\\  
    A_{i,1}^{[s]}\\ 
    \vdots\\
    A_{i,t_{i}-1}^{[s]}\\
    A_{i,t_{i}}^{[s]}\\
    A_{i,t_{i}+1}^{[s]} 
  \end{array}
\right)=
\left(                 
  \begin{array}{c}   
    u_{i,i_{1}}^{[s]}\\  
    u_{i,i_{2}}^{[s]}\\  
    \vdots\\ 
    u_{i,i_{t_{i}}}^{[s]}\\ 
    u_{i,n+1}^{[s]}\\ 
    u_{i,n+2}^{[s]} 
  \end{array}
\right),                 
$$
where $u_{i,j}^{[s]}$ ($j\in I_{i}'=I_{i}\cup \{n+1,n+2\}$) means the $s$th component of vector $u_{i,j}$, and $A_{i,v}^{[s]}$ ($0\leq v\leq t_{i}+1$) denotes corresponding coefficient in the general term of $u_{i,j}^{[s]}$ as in Theorem 1.

Then, they obtain $A_{i,0}^{[s]},A_{i,1}^{[s]},\cdots,A_{i,t_{i}+1}^{[s]}$ respectively, and further the general expression of $u_{i,j}^{[s]}$ in $[*]_{l=1}$ according to the Theorem 1:
$$u_{i,j}^{[s]}=p_{i}^{[s]}(j)=A_{i,0}^{[s]}+A_{i,1}^{[s]}j+\cdots+A_{i,t_{i}+1}^{[s]}j^{t_{i}+1}.$$
Consequently,
$$S_{i}^{[s]}=u_{i,0}^{[s]}=p_{i}^{[s]}(0)=A_{i,0}^{[s]}.$$
It can be seen from the above, these authorized participants can restore the secret $S_{i}$ through $A_{i,0}^{[s]}$ $(1\leq s\leq t_{i})$ directly.

Method 2. For $1\leq i\leq k$, utilizing $t_{i}$ pairs $\{(j,u_{i,j})|j \in I_{i}\}$ and published $(n+1,u_{i,n+1})$, $(n+2,u_{i,n+2})$, they can restore the shared secret $S_{i}$ by the Lagrange Interpolation Formula:
$$S_{i}^{[s]}=p_{i}^{[s]}(0)=\sum_{j\in I_{i}'}u_{i,j}^{[s]}\prod_{y\in I_{i}', y\neq j}\frac{y}{y-j}\quad (1\leq s\leq t_{i}),$$
where $I_{i}'=I_{i}\cup \{n+1,n+2\}=\{i_{1},i_{2},\cdots,i_{t_{i}}\}\cup \{n+1,n+2\}$.

$\textbf{Way\: 2}$: For $1\leq i\leq k$, we assume that $t_{i}$ participants $\{P_{j},P_{j+1},\cdots,P_{j+t_{i}-1}\}$ $(1\leq j\leq n-t_{i}+1)$ collaborate to reconstruct the secrets. They pool their shadows $\{d_{m}^{i}\}_{m\in I_{i}}$ $(I_{i}=\{j,j+1,\cdots,j+t_{i}-1\})$ to calculate $t_{i}$ successive subshadows as follows:
$$u_{i,m}=\left\{
\begin{aligned}
& d_{m}^{i},\qquad\quad\:\, if\quad 1\leq m\leq t_{i}-1;\\
& d_{m}^{i}+r_{i,m}, \quad if\quad t_{i}\leq m\leq n.
\end{aligned}
\right.$$

For $1\leq i\leq k$, since $c_{i}$ has been released, they can use the formulas below to calculate $u_{i,j-1},u_{i,j-2},$ $\cdots,u_{i,0}$ one by one:
$$u_{i,m}=(-1)^{t_{i}}c_{i}m-\sum_{\lambda=0}^{t_{i}-1}\left( {\begin{array}{*{20}{ccc}}
	t_{i}\\
	\lambda
	\end{array}} \right)(-1)^{\lambda+t_{i}}u_{i,m+t_{i}-\lambda}\qquad (0\leq m< j).$$
Consequently, they have $u_{i,0}=S_{i}$ for $1\leq i\leq k$.


\subsection{Verification phase}

After recovery of the secrets, these participants can use the values of $H(S_{i})$ $(1\leq i\leq k)$ which have been published on the bulletin board to test the correctness of the shared secrets.

\section{Scheme 2}
In this scheme, let $\mathcal{D}$ be the dealer, $\mathcal{P}=\{P_{1},P_{2},\cdots,P_{n}\}$ be the collection of participants, and $\mathcal{S}=\{S_{1},S_{2},\cdots,S_{k}\}$ be a set of secrets. $S_{i}$ is from $\mathbb{F}_{q}^{t_{i}}$ where $t_{i}$ ($0< t_{i}\leq n$) is the threshold corresponding to every $S_{i}$ $(1\leq i\leq k)$, and $q$ is a big prime number. It should be noted that all matrix operations are performed on $\mathbb{F}_{q}$. Especially, the $ILR$ relation used here is a special case of relation (5) which is a Type-$t$ $ILR$ relation, and we set $l=1$.

\subsection{Setup phase}
(1) $\mathcal{D}$ randomly selects $sh_{j}\in \{0,1\}^{r}$ $(j=1,2,\cdots,n)$ as $P_{j}$'s share such that $sh_{i}\neq sh_{j}$ for $i\neq j$ where $r\geq \max \{t\log t, \log n\}$ with $t=\max\{t_{1},t_{2},\cdots,t_{k}\}$, and distributes the share vector $sh_{j}$ to the participant $P_{j}$ via a secure channel.

(2) $\mathcal{D}$ chooses uniformly random matrices $G_{i}\in \mathbb{F}_{q}^{t_{i} \times r}$ for $i=1,2,\cdots,k$.

(3) $\mathcal{D}$ chooses a uniformly random matrix $F\in \mathbb{F}_{q}^{t\times r}$, and computes $h_{j}=F\cdot sh_{j}$ for $j=1,2,\cdots,n$ by using Ajtai's function.

(4) $\mathcal{D}$ computes $H(S_{i})$ for $i=1,2,\cdots,k$, where $H(\cdot)$ is any public one-way hash function.

(5) $\mathcal{D}$ publishes $G_{1},G_{2},\cdots,G_{k},F,h_{1},h_{2},\cdots,h_{n},H(S_{1}),H(S_{2}),\cdots,H(S_{k})$ on the bulletin board.

\subsection{Construction phase}

$\mathcal{D}$ generates the subshadow $u_{i,j}$ for each secret $S_{i}$ $(1\leq i\leq k)$ and each participant $P_{j}$ $(1\leq j\leq n)$, and two extra $u_{i,n+1}$, $u_{i,n+2}$ for each secret $S_{i}$ $(1\leq i\leq k)$ as follows:

(1) Randomly chooses $k$ different constant vectors $c_{i}\in \mathbb{F}_{q}^{t_{i}}$, $i=1,2,\cdots,k$.

(2) For $i=1,2,\cdots,k$, computes shadow $d_{j}^{i}=G_{i}\cdot sh_{j}$ for $j=1,2,\cdots, t_{i}-1$ by using Ajtai's function.

(3) For $i=1,2,\cdots,k$, considers an inhomogeneous linear recurring sequence $\{u_{i,j}\}_{j \geq 0}$ defined by the following formula:
$$[**]_{l=1}=\left\{
\begin{aligned}
& u_{i,0}=S_{i},u_{i,1}=d_{1}^{i},\cdots,u_{i,t_{i}-1}=d_{t_{i}-1}^{i},\\
& \sum_{\lambda=0}^{t_{i}}\left( {\begin{array}{*{20}{ccc}}
	t_{i}\\
	\lambda
	\end{array}} \right)u_{i,j+t_{i}-\lambda}=(-1)^{j}c_{i}j\quad (j\geq 0),
\end{aligned}
\right.$$
and calculates $u_{i,j}$ for $1\leq i\leq k$ and $t_{i}\leq j\leq n+2$.

(4) Calculates $r_{i,j}=u_{i,j}-d_{j}^{i}$ for $1\leq i\leq k$ and $t_{i}\leq j\leq n$.

(5) Publishes $c_{1},c_{2},\cdots,c_{k}$, $r_{i,j}$ for $1\leq i\leq k$, $t_{i} \leq j\leq n$, $u_{i,n+1}$ and $u_{i,n+2}$ for $1\leq i\leq k$.

\subsection{Recovery phase}

When the participant $P_{j}$ receives his share $sh_{j}$, $1\leq j\leq n$, he examines whether the following equation holds:
$$F\cdot sh_{j}\overset{?}{=}h_{j},\qquad 1\leq j\leq n.$$
If the shares are proved to be authentic, any authorized participants can recover the secrets with any order by the following two ways according to different conditions they meet.\\

$\textbf{Way\: 1}$: For $1\leq i\leq k$, assume that any $t_{i}$ participants $\{P_{j}\}_{j\in I_{i}}$ $(I_{i}=\{i_{1},i_{2},\cdots,i_{t_{i}}\}$ $\subseteq \{1,2,\cdots,n\})$ reconstruct the secret $S_{i}$ together. At first, these participants collect $t_{i}$ shadows $\{d_{j}^{i}\}_{j\in I_
{i}}$ to calculate $t_{i}$ subshadows as follows:
$$u_{i,j}=\left\{
\begin{aligned}
& d_{j}^{i},\qquad\quad\:\, if\quad 1\leq j\leq t_{i}-1;\\
& d_{j}^{i}+r_{i,j}, \quad if\quad t_{i}\leq j\leq n.
\end{aligned}
\right.$$

Method 1. They utilize $t_{i}$ subshadows $\{u_{i,j}|j\in I_{i}\}$ and published $\{u_{i,n+1}$, $u_{i,n+2}\}$ to solve the following Vandermond equations for $1\leq i\leq k$ and $1\leq s\leq t_{i}$:

$$      
\left(                 
  \begin{array}{ccccc}   
    1 & i_{1} & i_{1}^{2} \cdots & i_{1}^{t_{i}+1}\\  
    1 & i_{2} & i_{2}^{2} \cdots & i_{2}^{t_{i}+1}\\
    \vdots & \vdots & \vdots & \vdots \\
    1 & i_{t_{i}} & i_{t_{i}}^{2} \cdots & i_{t_{i}}^{t_{i}+1}\\ 
    1 & n+1 & (n+1)^{2} \cdots & (n+1)^{t_{i}+1} \\ 
    1 & n+2 & (n+2)^{2} \cdots & (n+2)^{t_{i}+1}  
  \end{array}
\right)\quad
\left(                 
  \begin{array}{c}   
    A_{i,0}^{[s]}\\  
    A_{i,1}^{[s]}\\ 
    \vdots\\
    A_{i,t_{i}-1}^{[s]}\\
    A_{i,t_{i}}^{[s]}\\
    A_{i,t_{i}+1}^{[s]} 
  \end{array}
\right)=
\left(                 
  \begin{array}{c}   
    u_{i,i_{1}}^{[s]}\\  
    u_{i,i_{2}}^{[s]}\\  
    \vdots\\ 
    u_{i,i_{t_{i}}}^{[s]}\\ 
    u_{i,n+1}^{[s]}\\ 
    u_{i,n+2}^{[s]} 
  \end{array}
\right),                 
$$
where $u_{i,j}^{[s]}$ ($j\in I_{i}'=I_{i}\cup \{n+1,n+2\}$) means the $s$th component of vector $u_{i,j}$, and $A_{i,v}^{[s]}$ ($0\leq v\leq t_{i}+1$) denotes corresponding coefficient in the general term of $u_{i,j}^{[s]}$ as in Theorem 2.

Then, they obtain $A_{i,0}^{[s]},A_{i,1}^{[s]},\cdots,A_{i,t_{i}+1}^{[s]}$ respectively, and further the general expression of $u_{i,j}^{[s]}$ in $[**]_{l=1}$ according to the Theorem 2:
$$u_{i,j}^{[s]}=(-1)^{j}p_{i}^{[s]}(j)=(-1)^{j}(A_{i,0}^{[s]}+A_{i,1}^{[s]}j+\cdots+A_{i,t_{i}+1}^{[s]}j^{t_{i}+1}).$$
Consequently,
$$S_{i}^{[s]}=u_{i,0}^{[s]}=(-1)^{0}p_{i}^{[s]}(0)=A_{i,0}^{[s]}.$$
It can be seen from the above, these authorized participants can restore the secret $S_{i}$ through $A_{i,0}^{[s]}$ $(1\leq s\leq t_{i})$ directly.

Method 2. For $1\leq i\leq k$, utilizing $t_{i}$ pairs $\{(j,u_{i,j})|j \in I_{i}\}$ and published $(n+1,u_{i,n+1})$, $(n+2,u_{i,n+2})$, they can restore the shared secret $S_{i}$ by the Lagrange Interpolation Formula:
$$S_{i}^{[s]}=(-1)^{0}p_{i}^{[s]}(0)=\sum_{j\in I_{i}'}u_{i,j}^{[s]}\prod_{y\in I_{i}', y\neq j}\frac{y}{y-j}\quad (1\leq s\leq t_{i}),$$
where $I_{i}'=I_{i}\cup \{n+1,n+2\}=\{i_{1},i_{2},\cdots,i_{t_{i}}\}\cup \{n+1,n+2\}$.

$\textbf{Way\: 2}$: For $1\leq i\leq k$, we assume that $t_{i}$ participants $\{P_{j},P_{j+1},\cdots,P_{j+t_{i}-1}\}$ $(1\leq j\leq n-t_{i}+1)$ collaborate to reconstruct the secrets. They pool their shadows $\{d_{m}^{i}\}_{m\in I_{i}}$ $(I_{i}=\{j,j+1,\cdots,j+t_{i}-1\})$ to calculate $t_{i}$ successive subshadows as follows:
$$u_{i,m}=\left\{
\begin{aligned}
& d_{m}^{i},\qquad\quad\:\, if\quad 1\leq m\leq t_{i}-1;\\
& d_{m}^{i}+r_{i,m}, \quad if\quad t_{i}\leq m\leq n.
\end{aligned}
\right.$$

For $1\leq i\leq k$, since $c_{i}$ has been released, they can use the formulas below to calculate $u_{i,j-1},u_{i,j-2},$ $\cdots,u_{i,0}$ one by one:
$$u_{i,m}=(-1)^{m}c_{i}m-\sum_{\lambda=0}^{t_{i}-1}\left( {\begin{array}{*{20}{ccc}}
	t_{i}\\
	\lambda
	\end{array}} \right)u_{i,m+t_{i}-\lambda}\quad (0\leq m< j).$$
Consequently, they have $u_{i,0}=S_{i}$ for $1\leq i\leq k$.

\subsection{Verification phase}

After recovery of the secrets, these participants can use the values of $H(S_{i})$ $(1\leq i\leq k)$ which have been published on the bulletin board to test the correctness of the shared secrets.

\section{Scheme 3}

In this scheme, let $\mathcal{D}$ be the dealer, $\mathcal{P}=\{P_{1},P_{2},\cdots,P_{n}\}$ be the collection of participants, and $\mathcal{S}=\{S_{1},S_{2},\cdots,S_{k}\}$ be a set of secrets. $S_{i}$ is from $\mathbb{F}_{q}^{l_{i}}$ where $l_{i}$ ($0< l_{i}\leq n$) is the threshold corresponding to every $S_{i}$ $(1\leq i\leq k)$, and $q$ is a big prime number. It should be noted that all matrix operations are performed on $\mathbb{F}_{q}$. Especially, the $ILR$ relation used here is a special case of relation (4) which is a Type-$l$ $ILR$ relation, and we set $t=1$.

\subsection{Setup phase}

(1) $\mathcal{D}$ randomly selects $sh_{j}\in \{0,1\}^{r}$ $(j=1,2,\cdots,n)$ as $P_{j}$'s share such that $sh_{i}\neq sh_{j}$ for $i\neq j$ where $r\geq \max (l\log l, \log n)$ with $l=\max\{l_{1},l_{2},\cdots,l_{k}\}$, and distributes the share vector $sh_{j}$ to the participant $P_{j}$ via a secure channel.

(2) $\mathcal{D}$ chooses uniformly random matrices $G_{i}\in \mathbb{F}_{q}^{l_{i} \times r}$ for $i=1,2,\cdots,k$.

(3) $\mathcal{D}$ chooses a uniformly random matrix $F\in \mathbb{F}_{q}^{l\times r}$, and computes $h_{j}=F\cdot sh_{j}$ for $j=1,2,\cdots,n$ by using Ajtai's function.

(4) $\mathcal{D}$ computes $H(S_{i})$ for $i=1,2,\cdots,k$, where $H(\cdot)$ is any public one-way hash function.

(5) $\mathcal{D}$ publishes $G_{1},G_{2},\cdots,G_{k},F,h_{1},h_{2},\cdots,h_{n},H(S_{1}),H(S_{2}),\cdots,H(S_{k})$ on the bulletin board.

\subsection{Construction phase}
$\mathcal{D}$ generates the subshadow $u_{i,j}$ for each secret $S_{i}$ $(1\leq i\leq k)$ and each participant $P_{j}$ $(1\leq j\leq n)$, and extra $u_{i,n+1}, u_{i,n+2}, \cdots, u_{i,n+l_{i}+1}$ for each secret $S_{i}$ $(1\leq i\leq k)$ as follows:

(1) Randomly chooses a constant vector $c\in \mathbb{F}_{q}^{l}$, where $l=\max \{l_{1},l_{2},\cdots,l_{k}\}$.

(2) For $i=1,2,\cdots,k$, computes shadow $d_{j}^{i}=G_{i}\cdot sh_{j}$ for $j=1,2,\cdots, l_{i}-1$ by using Ajtai's function.

(3) For $i=1,2,\cdots,k$, considers an inhomogeneous linear recurring sequence $\{u_{i,j}\}_{j \geq 0}$ defined by the following formula:
$$[***]_{t=1}=\left\{
\begin{aligned}
& u_{i,0}=S_{i},u_{i,1}=d_{1}^{i},\cdots,u_{i,l_{i}-1}=d_{l_{i}-1}^{i},\\
& \sum_{\lambda=0}^{l_{i}}\left( {\begin{array}{*{20}{ccc}}
	l_{i}\\
	\lambda
	\end{array}} \right)(-1)^{\lambda}u_{i,j+l_{i}-\lambda}=\frac{(j)_{l_{i}}}{l_{i}!}c_{[l_{i}]}\quad (j\geq 0),
\end{aligned}
\right.$$
where $c_{[l_{i}]}$ denotes the first $l_{i}$ components of the vector $c$, and calculates $u_{i,j}$ for $1\leq i\leq k$ and $l_{i}\leq j\leq n+l_{i}+1$.

(4) Calculates $r_{i,j}=u_{i,j}-d_{j}^{i}$ for $1\leq i\leq k$ and $l_{i}\leq j\leq n$.

(5) Publishes $c$, $r_{i,j}$ for $1\leq i\leq k$, $l_{i} \leq j\leq n$, $u_{i,n+1}, u_{i,n+2}, \cdots, u_{i,n+l_{i}+1}$ for $1\leq i\leq k$.

\subsection{Recovery phase}
When the participant $P_{j}$ receives his share $sh_{j}$, $1\leq j\leq n$, he examines whether the following equation holds:
$$F\cdot sh_{j}\overset{?}{=}h_{j},\qquad 1\leq j\leq n.$$
If the shares are proved to be authentic, any authorized participants can recover the secrets with any order by the following two ways according to different conditions they meet.\\

$\textbf{Way\: 1}$: For $1\leq i\leq k$, assume that any $l_{i}$ participants $\{P_{j}\}_{j\in I_{i}}$ $(I_{i}=\{i_{1},i_{2},\cdots,i_{l_{i}}\}$ $\subseteq \{1,2,\cdots,n\})$ reconstruct the secret $S_{i}$ together. At first, these participants collect $l_{i}$ shadows $\{d_{j}^{i}\}_{j\in I_
{i}}$ to calculate $l_{i}$ subshadows as follows:
$$u_{i,j}=\left\{
\begin{aligned}
& d_{j}^{i},\qquad\quad\:\, if\quad 1\leq j\leq l_{i}-1;\\
& d_{j}^{i}+r_{i,j}, \quad if\quad l_{i}\leq j\leq n.
\end{aligned}
\right.$$

Method 1. They utilize $l_{i}$ subshadows $\{u_{i,j}|j\in I_{i}\}$ and other $l_{i}+1$ terms $\{u_{i,n+1}$, $u_{i,n+2},\cdots,u_{i,n+l_{i}+1}\}$ to solve the following Vandermond equations for $1\leq i\leq k$ and $1\leq s\leq l_{i}$:

\begin{footnotesize}
$$     
\left(                 
  \begin{array}{ccccc}   
    1 & i_{1} & i_{1}^{2} \cdots & i_{1}^{2l_{i}}\\  
    1 & i_{2} & i_{2}^{2} \cdots & i_{2}^{2l_{i}}\\
    \vdots & \vdots & \vdots & \vdots \\
    1 & i_{l_{i}} & i_{l_{i}}^{2} \cdots & i_{l_{i}}^{2l_{i}}\\ 
    1 & n+1 & (n+1)^{2} \cdots & (n+1)^{2l_{i}} \\ 
    1 & n+2 & (n+2)^{2} \cdots & (n+2)^{2l_{i}} \\ 
     \vdots & \vdots & \vdots & \vdots \\
    1 & n+l_{i}+1 & (n+l_{i}+1)^{2} \cdots & (n+l_{i}+1)^{2l_{i}} \\ 
  \end{array}
\right)\quad
\left(                 
  \begin{array}{c}   
    A_{i,0}^{[s]}\\  
    A_{i,1}^{[s]}\\ 
    \vdots\\
    A_{i,l_{i}-1}^{[s]}\\ 
    A_{i,l_{i}}^{[s]}\\
    A_{i,l_{i}+1}^{[s]}\\ 
    \vdots\\
    A_{i,2l_{i}}^{[s]}
  \end{array}
\right)=
\left(                 
  \begin{array}{c}   
    u_{i,i_{1}}^{[s]}\\  
    u_{i,i_{2}}^{[s]}\\  
    \vdots\\ 
    u_{i,i_{l_{i}}}^{[s]}\\ 
    u_{i,n+1}^{[s]}\\ 
    u_{i,n+2}^{[s]}\\ 
    \vdots\\ 
    u_{i,n+l_{i}+1}^{[s]} 
  \end{array}
\right),                 
$$
\end{footnotesize}

\noindent where $u_{i,j}^{[s]}$ ($j\in I_{i}'=I_{i}\cup \{n+1,n+2,\cdots,n+l_{i}+1\}$) means the $s$th component of vector $u_{i,j}$, and $A_{i,v}^{[s]}$ ($0\leq v\leq 2l_{i}$) denotes corresponding coefficient in the general term of $u_{i,j}^{[s]}$ as in Theorem 1.

Then, they obtain $A_{i,0}^{[s]},A_{i,1}^{[s]},\cdots,A_{i,2l_{i}}^{[s]}$ respectively, and further the general expression of $u_{i,j}^{[s]}$ in $[***]_{t=1}$ according to the Theorem 1:

$$u_{i,j}^{[s]}=p_{i}^{[s]}(j)=A_{i,0}^{[s]}+A_{i,1}^{[s]}j+\cdots+A_{i,2l_{i}}^{[s]}j^{2l_{i}}.$$
Consequently,
$$S_{i}^{[s]}=u_{i,0}^{[s]}=p_{i}^{[s]}(0)=A_{i,0}^{[s]}.$$
It can be seen from the above, these authorized participants can restore the secret $S_{i}$ through $A_{i,0}^{[s]}$ $(1\leq s\leq l_{i})$ directly.

Method 2. For $1\leq i\leq k$, utilizing $2l_{i}+1$ pairs $\{(j,u_{i,j})|j \in I_{i}\}$ and published $\{(n+j,u_{i,n+j})|j=1,2,\cdots,l_{i}+1\}$, they can restore the shared secret $S_{i}$ by the Lagrange Interpolation Formula:
$$S_{i}^{[s]}=p_{i}^{[s]}(0)=\sum_{j\in I_{i}'}u_{i,j}^{[s]}\prod_{y\in I_{i}', y\neq j}\frac{y}{y-j}\quad (1\leq s\leq l_{i}),$$
where $I_{i}'=I_{i}\cup \{n+1,n+2,\cdots,n+l_{i}+1\}=\{i_{1},i_{2},\cdots,i_{l_{i}}\}\cup \{n+1,n+2,\cdots,n+l_{i}+1\}$.

$\textbf{Way\: 2}$: For $1\leq i\leq k$, we assume that $l_{i}$ participants $\{P_{j},P_{j+1},\cdots,P_{j+l_{i}-1}\}$ $(1\leq j\leq n-l_{i}+1)$ collaborate to reconstruct the secrets. They pool their shadows $\{d_{m}^{i}\}_{m\in I_{i}}$ $(I_{i}=\{j,j+1,\cdots,j+l_{i}-1\})$ to calculate $l_{i}$ successive subshadows as follows:
$$u_{i,m}=\left\{
\begin{aligned}
& d_{m}^{i},\qquad\quad\:\, if\quad 1\leq m\leq l_{i}-1;\\
& d_{m}^{i}+r_{i,m},\quad if\quad l_{i}\leq m\leq n.
\end{aligned}
\right.$$

For $1\leq i\leq k$, since $c$ has been released, they can use the formulas below to calculate $u_{i,j-1},u_{i,j-2},$ $\cdots,u_{i,0}$ one by one:
$$u_{i,m}=(-1)^{l_{i}}\frac{(m)_{l_{i}}}{l_{i}!}c_{[l_{i}]}-\sum_{\lambda=0}^{l_{i}-1}\left( {\begin{array}{*{20}{ccc}}
	l_{i}\\
	\lambda
	\end{array}} \right)(-1)^{\lambda+l_{i}}u_{i,m+l_{i}-\lambda}\qquad (0\leq m< j).$$
Consequently, they have $u_{i,0}=S_{i}$ for $1\leq i\leq k$.

\subsection{Verification phase}
After recovery of the secrets, these participants can use the values of $H(S_{i})$ $(1\leq i\leq k)$ which have been published on the bulletin board to test the correctness of the shared secrets.

\section{Scheme 4}
In this scheme, let $\mathcal{D}$ be the dealer, $\mathcal{P}=\{P_{1},P_{2},\cdots,P_{n}\}$ be the collection of participants, and $\mathcal{S}=\{S_{1},S_{2},\cdots,S_{k}\}$ be a set of secrets. $S_{i}$ is from $\mathbb{F}_{q}^{l_{i}}$ where $l_{i}$ ($0< l_{i}\leq n$) is the threshold corresponding to every $S_{i}$ $(1\leq i\leq k)$, and $q$ is a big prime number. It should be noted that all matrix operations are performed on $\mathbb{F}_{q}$. Especially, the $ILR$ relation used here is a special case of relation (5) which is a Type-$l$ $ILR$, and we set $t=1$.

\subsection{Setup phase} 
(1) $\mathcal{D}$ randomly selects $sh_{j}\in \{0,1\}^{r}$ $(j=1,2,\cdots,n)$ as $P_{j}$'s share such that $sh_{i}\neq sh_{j}$ for $i\neq j$ where $r\geq \max (l\log l, \log n)$ with $l=\max\{l_{1},l_{2},\cdots,l_{k}\}$, and distributes the share vector $sh_{j}$ to the participant $P_{j}$ via a secure channel.

(2) $\mathcal{D}$ chooses uniformly random matrices $G_{i}\in \mathbb{F}_{q}^{l_{i} \times r}$ for $i=1,2,\cdots,k$.

(3) $\mathcal{D}$ chooses a uniformly random matrix $F\in \mathbb{F}_{q}^{l\times r}$, and computes $h_{j}=F\cdot sh_{j}$ for $j=1,2,\cdots,n$ by using Ajtai's function.

(4) $\mathcal{D}$ computes $H(S_{i})$ for $i=1,2,\cdots,k$, where $H(\cdot)$ is any public one-way hash function.

(5) $\mathcal{D}$ publishes $G_{1},G_{2},\cdots,G_{k},F,h_{1},h_{2},\cdots,h_{n},H(S_{1}),H(S_{2}),\cdots,H(S_{k})$ on the bulletin board.

\subsection{Construction phase} 

$\mathcal{D}$ generates the subshadow $u_{i,j}$ for each secret $S_{i}$ $(1\leq i\leq k)$ and each participant $P_{j}$ $(1\leq j\leq n)$, and extra $u_{i,n+1}, u_{i,n+2}, \cdots, u_{i,n+l_{i}+1}$ for each secret $S_{i}$ $(1\leq i\leq k)$ as follows:

(1) Randomly chooses a constant vector $c\in \mathbb{F}_{q}^{l}$, where $l=\max \{l_{1},l_{2},\cdots,l_{k}\}$.

(2) For $i=1,2,\cdots,k$, computes shadow $d_{j}^{i}=G_{i}\cdot sh_{j}$ for $j=1,2,\cdots, l_{i}-1$ by using Ajtai's function.

(3) For $i=1,2,\cdots,k$, considers an inhomogeneous linear recurring sequence $\{u_{i,j}\}_{j\geq 0}$ defined by the following formula:
$$[****]_{t=1}=\left\{
\begin{aligned}
& u_{i,0}=S_{i},u_{i,1}=d_{1}^{i},\cdots,u_{i,l_{i}-1}=d_{l_{i}-1}^{i},\\
& \sum_{\lambda=0}^{l_{i}}\left( {\begin{array}{*{20}{ccc}}
	l_{i}\\
	\lambda
	\end{array}} \right)u_{i,j+l_{i}-\lambda}=(-1)^{j}\frac{(j)_{l_{i}}}{l_{i}!}c_{[l_{i}]}\quad (j\geq 0),
\end{aligned}
\right.$$
where $c_{[l_{i}]}$ denotes the first $l_{i}$ components of the vector $c$, and calculates $u_{i,j}$ for $1\leq i\leq k$ and $l_{i}\leq j\leq n+l_{i}+1$.

(4) Calculates $r_{i,j}=u_{i,j}-d_{j}^{i}$ for $1\leq i\leq k$ and $l_{i}\leq j\leq n$.

(5) Publishes $c$, $r_{i,j}$ for $1\leq i\leq k$, $l_{i} \leq j\leq n$, $u_{i,n+1}, u_{i,n+2}, \cdots, u_{i,n+l_{i}+1}$ for $1\leq i\leq k$.

\subsection{Recovery phase} 

When the participant $P_{j}$ receives his share $sh_{j}$, $1\leq j\leq n$, he examines whether the following equation holds:
$$F\cdot sh_{j}\overset{?}{=}h_{j},\qquad 1\leq j\leq n.$$
If the shares are proved to be authentic, any authorized participants can recover the secrets with any order by the following two ways according to different conditions they meet.\\

$\textbf{Way\: 1}$: For $1\leq i\leq k$, assume that any $l_{i}$ participants $\{P_{j}\}_{j\in I_{i}}$ $(I_{i}=\{i_{1},i_{2},\cdots,i_{l_{i}}\}$ $\subseteq \{1,2,\cdots,n\})$ reconstruct the secret $S_{i}$ together. At first, these participants collect $l_{i}$ shadows $\{d_{j}^{i}\}_{j\in I_
{i}}$ to calculate $l_{i}$ subshadows as follows:
$$u_{i,j}=\left\{
\begin{aligned}
& d_{j}^{i},\qquad\quad\:\, if\quad 1\leq j\leq l_{i}-1;\\
& d_{j}^{i}+r_{i,j},\quad if\quad l_{i}\leq j\leq n.
\end{aligned}
\right.$$

Method 1. They utilize $l_{i}$ subshadows $\{u_{i,j}|j\in I_{i}\}$ and other $l_{i}+1$ terms $\{u_{i,n+1}$, $u_{i,n+2},\cdots,u_{i,n+l_{i}+1}\}$ to solve the following Vandermond equations for $1\leq i\leq k$ and $1\leq s\leq l_{i}$:

\begin{footnotesize}
$$    
\left(                 
  \begin{array}{ccccc}   
    1 & i_{1} & i_{1}^{2} \cdots & i_{1}^{2l_{i}}\\  
    1 & i_{2} & i_{2}^{2} \cdots & i_{2}^{2l_{i}}\\
    \vdots & \vdots & \vdots & \vdots \\
    1 & i_{l_{i}} & i_{l_{i}}^{2} \cdots & i_{l_{i}}^{2l_{i}}\\ 
    1 & n+1 & (n+1)^{2} \cdots & (n+1)^{2l_{i}} \\ 
    1 & n+2 & (n+2)^{2} \cdots & (n+2)^{2l_{i}} \\ 
     \vdots & \vdots & \vdots & \vdots \\
    1 & n+l_{i}+1 & (n+l_{i}+1)^{2} \cdots & (n+l_{i}+1)^{2l_{i}} \\ 
  \end{array}
\right)\quad
\left(                 
  \begin{array}{c}   
    A_{i,0}^{[s]}\\  
    A_{i,1}^{[s]}\\ 
    \vdots\\
    A_{i,l_{i}-1}^{[s]}\\ 
    A_{i,l_{i}}^{[s]}\\
    A_{i,l_{i}+1}^{[s]}\\ 
    \vdots\\
    A_{i,2l_{i}}^{[s]}
  \end{array}
\right)=
\left(                 
  \begin{array}{c}   
    u_{i,i_{1}}^{[s]}\\  
    u_{i,i_{2}}^{[s]}\\  
    \vdots\\ 
    u_{i,i_{l_{i}}}^{[s]}\\ 
    u_{i,n+1}^{[s]}\\ 
    u_{i,n+2}^{[s]}\\ 
    \vdots\\ 
    u_{i,n+l_{i}+1}^{[s]} 
  \end{array}
\right),                 
$$
\end{footnotesize}

\noindent where $u_{i,j}^{[s]}$ ($j\in I_{i}'=I_{i}\cup \{n+1,n+2,\cdots,n+l_{i}+1\}$) means the $s$th component of vector $u_{i,j}$, and $A_{i,v}^{[s]}$ ($0\leq v\leq 2l_{i}$) denotes corresponding coefficient in the general term of $u_{i,j}^{[s]}$ as in Theorem 2.

Then, they obtain $A_{i,0}^{[s]},A_{i,1}^{[s]},\cdots,A_{i,2l_{i}}^{[s]}$ respectively, and further the general expression of $u_{i,j}^{[s]}$ in $[****]_{t=1}$ according to the Theorem 2:
$$u_{i,j}^{[s]}=(-1)^{j}p_{i}^{[s]}(j)=(-1)^{j}(A_{i,0}^{[s]}+A_{i,1}^{[s]}j+\cdots+A_{i,2l_{i}}^{[s]}j^{2l_{i}}).$$
Consequently,
$$S_{i}^{[s]}=u_{i,0}^{[s]}=(-1)^{0}p_{i}^{[s]}(0)=A_{i,0}^{[s]}.$$
It can be seen from the above, these authorized participants can restore the secret $S_{i}$ through $A_{i,0}^{[s]}$ $(1\leq s\leq l_{i})$ directly.

Method 2. For $1\leq i\leq k$, utilizing $2l_{i}+1$ pairs $\{(j,u_{i,j})|j \in I_{i}\}$ and published $\{(n+j,u_{i,n+j})|j=1,2,\cdots,l_{i}+1\}$, they can restore the shared secret $S_{i}$ by the Lagrange Interpolation Formula:
$$S_{i}^{[s]}=(-1)^{0}p_{i}^{[s]}(0)=\sum_{j\in I_{i}'}u_{i,j}^{[s]}\prod_{y\in I_{i}', y\neq j}\frac{y}{y-j}\quad (1\leq s\leq l_{i}),$$
where $I_{i}'=I_{i}\cup \{n+1,n+2,\cdots,n+l_{i}+1\}=\{i_{1},i_{2},\cdots,i_{l_{i}}\}\cup \{n+1,n+2,\cdots,n+l_{i}+1\}$.

$\textbf{Way\: 2}$: For $1\leq i\leq k$, we assume that $l_{i}$ participants $\{P_{j},P_{j+1},\cdots,P_{j+l_{i}-1}\}$ $(1\leq j\leq n-l_{i}+1)$ collaborate to reconstruct the secrets. They pool their shadows $\{d_{m}^{i}\}_{m\in I_{i}}$ $(I_{i}=\{j,j+1,\cdots,j+l_{i}-1\})$ to calculate $l_{i}$ successive subshadows as follows:
$$u_{i,m}=\left\{
\begin{aligned}
& d_{m}^{i},\qquad\quad\:\, if\quad 1\leq m\leq l_{i}-1;\\
& d_{m}^{i}+r_{i,m},\quad if\quad l_{i}\leq m\leq n.
\end{aligned}
\right.$$

For $1\leq i\leq k$, since $c$ has been released, they can use the formulas below to calculate $u_{i,j-1},u_{i,j-2},$ $\cdots,u_{i,0}$ one by one:
$$u_{i,m}=(-1)^{m}\frac{(m)_{l_{i}}}{l_{i}!}c_{[l_{i}]}-\sum_{\lambda=0}^{l_{i}-1}\left( {\begin{array}{*{20}{ccc}}
	l_{i}\\
	\lambda
	\end{array}} \right)u_{i,m+l_{i}-\lambda}\qquad (0\leq m< j).$$
Consequently, they have $u_{i,0}=S_{i}$ for $1\leq i\leq k$.
\subsection{Verification phase}
After recovery of the secrets, these participants can use the values of $H(S_{i})$ $(1\leq i\leq k)$ which have been published on the bulletin board to test the correctness of the shared secrets.

\section{Security analysis}

According to the definition of the security, then we prove the security of our new MSSSAOs.

\subsection{Correctness}

If the dealer and participants behave honestly, utilizing the two ways mentioned in the recovery phase of our schemes, for $1\leq i\leq k$, any $t_{i}$ or $l_{i}$ participants can restore corresponding shared secret $S_{i}$ independently.

\subsection{Verifiability}

From the description of the verification phase, our schemes can verify the correctness of recovered secrets easily making full use of the property of hash function.

\subsection{Privacy}
Next, we will show that the adversary cannot get any information about the unrecovered secrets from corresponding shadows of the recovered secrets and the public messages, respectively. Finally, we prove that any number of participants less than the threshold cannot restore the secrets. Notice that we let the threshold corresponding to $S_{i}$ be $t_{i}$ for convenience in the discussion below. When the threshold is $l_{i}$, the conclusions below are similar.

Before we provide the first consequence, we give three lemmas at first.

\begin{lemma}
Let $A$ be a uniformly random matrix selected from $\mathbb{F}_{q}^{t\times r}$. Denote $\mathfrak{B}$ the set of invertible matrices from $\mathbb{F}_{q}^{r\times r}$. Let $B$ be a random matrix selected from $\mathfrak{B}$. Then, $C=A\cdot B$ is a uniformly random matrix from $\mathbb{F}_{q}^{t\times r}$.
\end{lemma}

\begin{proof}
Given $C=A \cdot B$, we get the following consequence:
$$Pr\{C=c\}=Pr\{A\cdot B=c\}=\sum_{b\in \mathfrak{B}}Pr\{A=c\cdot b^{-1}\}\cdot Pr\{B=b\}$$
$$=\frac{1}{q^{t\times r}}\sum_{b\in \mathfrak{B}}Pr\{B=b\}=\frac{1}{q^{t\times r}},$$
where $b^{-1}$ denotes the inverse matrix of the matrix $b$. Therefore, $C=A\cdot B$ is a uniformly random matrix from $\mathbb{F}_{q}^{t\times r}$.

\end{proof}

\begin{lemma}
Let $G$ be a uniformly random matrix chosen from $\mathbb{F}_{q}^{t\times r}$ ($r\geq t\log t$), then the matrix $G$ has full row rank with a non-negligible probability. More specifically, the probability of a uniformly random matrix $G\in \mathbb{F}_{q}^{t\times r}$ ($r\geq t\log t$) without full row rank can be expressed as a negligible function of $t$.

\end{lemma}

\begin{proof}
From the formula in \cite{MP13}, we know that for a uniformly random matrix $G$ selected from $\mathbb{F}_{q}^{t\times r}$ ($r\geq t\log t$), the probability of this matrix with full row rank is:
$$Pr(t,r)=\prod_{j=1}^{t}(1-q^{-(r+1-j)})>1-\sum_{j=1}^{t}q^{-(r+1-j)}$$
$$\geq 1-\sum_{j=1}^{t}q^{-(t\log t+1-j)}.$$
The last inequality holds because $r\geq t\log t$.

Hence, the matrix $G$ without full row rank has the probability as follows:
$$1-Pr(t,r)\leq \sum_{j=1}^{t}q^{-(t\log t+1-j)}\approx q^{-(t\log t+1-t)}
=negl(t),$$
where $negl(t)$ is a negligible function of $t$.

\end{proof}

\begin{lemma}
$G_{1}$ and $G_{2}$ are two uniformly random matrices selected from $\mathbb{F}_{q}^{t_{1}\times r}$ and $\mathbb{F}_{q}^{t_{2}\times r}$, where $r\geq \max\{t_{1}\log t_{1},t_{2}\log t_{2}\}$, and let $\mathbf{y}$ be randomly distributed on $\{0,1\}^{r}$. Then it is impossible for any adversary to get $G_{2}\mathbf{y}$ from $G_{1}\mathbf{y}$, or get $G_{1}\mathbf{y}$ from $G_{2}\mathbf{y}$ in polynomial time.

\end{lemma}

\begin{proof}
The lemma is proved by reduction to absurdity. Suppose that there is an algorithm $\mathcal{A}$ which makes it possible for any adversary to get $G_{2}\mathbf{y}$ from $G_{1}\mathbf{y}$ in polynomial time with a non-negligible probability. Let $\mathcal{B}$ denote another algorithm which can invert the Ajtai's function with a non-negligible probability. From the conclusion in \cite{A96}, if this algorithm $\mathcal{B}$ exists, it means that the adversary can solve an $n^{c}$-approximate SVP with a non-negligible probability, which is impossible with today's computing power. Therefore, in order to prove the conclusion, we just need to prove that the adversary can perform the algorithm $\mathcal{B}$ by utilizing the algorithm $\mathcal{A}$.

Firstly, a uniformly random matrix $G_{2}$ is selected from $\mathbb{F}_{q}^{t_{2}\times r}$ by $\mathcal{B}$. Then, $\mathcal{B}$ can obtain a matrix $H_{2}\in \mathbb{F}_{q}^{r\times (r-t_{2})}$ satisfying the equation $G_{2}H_{2}=0_{t_{2}\times (r-t_{2})}$. Later, $\mathcal{B}$ calculates $G_{1}=E_{1}[G_{2}^{\dagger}\: H_{2}]^{-1}$, where $E_{1}$ is a uniformly random matrix from $\mathbb{F}_{q}^{t_{1}\times r}$ and $G_{2}^{\dagger}\in \mathbb{F}_{q}^{r\times t_{2}}$ is the pseudo-inverse of $G_{2}$.

Secondly, we will prove that $G_{1}$ is also a uniformly random matrix from $\mathbb{F}_{q}^{t_{1}\times r}$. According to Lemma 3, we know that the probability of $G_{2}$ with full row rank is non-negligible. Hence, $G_{2}^{\dagger}=G_{2}^{T}(G_{2}G_{2}^{T})^{-1}$. Next, we show that $[G_{2}^{\dagger}\: H_{2}]$ is invertible. Primarily, $G_{2}^{\dagger}$ and $H_{2}$ are matrices with full column rank. Besides, the columns of $G_{2}^{\dagger}$ cannot be written in the form of the linear representation of the columns of $H_{2}$, which is because that $G_{2}G_{2}^{\dagger}=I_{t_{2}}$ and $G_{2}H_{2}=0_{t_{2}\times (r-t_{2})}$, where $I_{t_{2}}$ denotes the identity matrix of order $t_{2}$. Thus, $[G_{2}^{\dagger}\: H_{2}]$ is invertible. Since $E_{1}$ is a uniformly random matrix, then the product $G_{1}=E_{1}[G_{2}^{\dagger}\: H_{2}]^{-1}$ is also a uniformly random matrix from $\mathbb{F}_{q}^{t_{1}\times r}$ due to Lemma 2.

Thirdly, we will prove that if the algorithm $\mathcal{A}$ exists, then the algorithm $\mathcal{B}$ exists. Let $\mathbf{y}=[G_{2}^{\dagger}\: H_{2}]\mathbf{x}$, where $\mathbf{x}\in \{0,1\}^{r}$ is a random vector. Since $G_{1}=E_{1}[G_{2}^{\dagger}\: H_{2}]^{-1}$, then $E_{1}=G_{1}\cdot [G_{2}^{\dagger}\: H_{2}]$ and therefore $E_{1}\mathbf{x}=G_{1}\cdot [G_{2}^{\dagger}\: H_{2}]\mathbf{x}=G_{1}\mathbf{y}$. Besides, we know that the algorithm $\mathcal{A}$ can output $G_{2}\mathbf{y}$ based on input $G_{1}\mathbf{y}$ in polynomial time with a non-negligible probability. Consequently, input $G_{1}\mathbf{y}$, i.e., $E_{1}\mathbf{x}$ and $E_{1}$, and the algorithm $\mathcal{B}$ can output $G_{2}\mathbf{y}$ with the help of the algorithm $\mathcal{A}$. And in fact, $G_{2}\mathbf{y}$ satisfies the following equation:
$$G_{2}\mathbf{y}=G_{2}[G_{2}^{\dagger}\: H_{2}]\mathbf{x}=[I_{t_{2}}\: 0_{t_{2}\times (r-t_{2})}]\mathbf{x}=\mathbf{x_{1}},$$
where $\mathbf{x}=\left[                 
  \begin{array}{c}   
  \mathbf{x_{1}}_{t_{2}\times 1}\\  
  \mathbf{x_{2}}_{(r-t_{2})\times 1}
  \end{array}
\right]$.
In other words, if the adversary takes $E_{1}\mathbf{x}$ and $E_{1}$ as inputs, he can obtain the first part of $\mathbf{x}$ by performing the algorithm $\mathcal{B}$. That is to say, the adversary can invert the Ajtai's function, in contradiction to the the conclusion in \cite{A96} as we explained at the beginning.

In conclusion, the algorithm $\mathcal{A}$ does not exist. Therefore, it is impossible for any adversary to get $G_{2}\mathbf{y}$ from $G_{1}\mathbf{y}$ in polynomial time. Similarly, it is also impossible to get $G_{1}\mathbf{y}$ from $G_{2}\mathbf{y}$ in polynomial time.

\end{proof}

\begin{theorem}

The adversary cannot get any information about the unrecovered secrets from corresponding shadows of the recovered secrets with a non-negligible probability.

\end{theorem}

\begin{proof}
Combining Lemma 4 with our schemes, for $1\leq i\neq j \leq k$ and $1\leq l\leq n$, it is impossible for the adversary to get $G_{i}sh_{l}$ from $G_{j}sh_{l}$ in polynomial time where $G_{i}$ and $G_{j}$ are any two uniformly random matrices selected from $\mathbb{F}_{q}^{t_{i}\times r}$ and $\mathbb{F}_{q}^{t_{j}\times r}$ independently, and $sh_{l}\in \{0,1\}^{r}$ is the share of $P_{l}$ with $r\geq \max \{t_{i}\log t_{i},t_{j}\log t_{j}, \log n\}$. That is to say, the shadow of the secret $S_{i}$ cannot be obtained by the shadow of another different secret $S_{j}$. Therefore, the adversary cannot get any information about the unrecovered secrets from the corresponding shadows of recovered secrets with a non-negligible probability.

\end{proof}

By Theorem 3, in our proposed schemes, we know that the recovered secrets cannot leak any information of the unrecovered secrets with a non-negligible probability. Therefore, any authorized sets can restore these secrets independently with computational security, which means that our schemes are computationally secure multi-stage secret sharing schemes with any order (MSSSAO).

Finally, we have the following theorem.

\begin{theorem}
Any number of participants less than the threshold $t_{i}$ cannot restore the unrevealed secret $S_{i}$ for $1\leq i\leq k$.
\end{theorem}

\begin{proof}
Without loss of generality, for $1\leq i\leq k$, we assume that $t_{i}-1$ participants conspire to reconstruct the unrevealed secret $S_{i}$ using their shadows and public messages in the $i$th-stage. On the one hand, according to Theorem 3, shadows corresponding to the secret $S_{i}$ cannot be computed from the shadows of revealed secrets in polynomial time with a non-negligible probability due to the one-wayness of Ajtai's function. On the other hand, to recover $S_{i}$, any $t_{i}-1$ participants cannot obtain $A_{i,0}\in \mathbb{F}_{q}^{t_{i}}$, which is because they cannot solve the linear system with $t_{i}(t_{i}+1)$ equations and $t_{i}(t_{i}+2)$ unknowns, i.e., for $1\leq s\leq t_{i}$,
$$      
\left(                 
  \begin{array}{ccccc}   
    1 & i_{1} & i_{1}^{2} \cdots & i_{1}^{t_{i}+1}\\  
    1 & i_{2} & i_{2}^{2} \cdots & i_{2}^{t_{i}+1}\\
    \vdots & \vdots & \vdots & \vdots \\
    1 & i_{t_{i}-1} & i_{t_{i}-1}^{2} \cdots & i_{t_{i}-1}^{t_{i}+1}\\ 
    1 & n+1 & (n+1)^{2} \cdots & (n+1)^{t_{i}+1} \\ 
    1 & n+2 & (n+2)^{2} \cdots & (n+2)^{t_{i}+1}  
  \end{array}
\right)\quad
\left(                 
  \begin{array}{c}   
    A_{i,0}^{[s]}\\  
    A_{i,1}^{[s]}\\ 
    \vdots\\
    A_{i,t_{i}-1}^{[s]}\\
    A_{i,t_{i}}^{[s]}\\
    A_{i,t_{i}+1}^{[s]} 
  \end{array}
\right)=
\left(                 
  \begin{array}{c}   
    u_{i,i_{1}}^{[s]}\\  
    u_{i,i_{2}}^{[s]}\\  
    \vdots\\ 
    u_{i,i_{t_{i}-1}}^{[s]}\\ 
    u_{i,n+1}^{[s]}\\ 
    u_{i,n+2}^{[s]} 
  \end{array}
\right),                 
$$
where there are $t_{i}$ degrees of freedom. Since $S_{i}=(A_{i,0}^{[1]},A_{i,0}^{[2]},\cdots,A_{i,0}^{[t_{i}]})^{T}$, we can regard these $t_{i}$ components of $S_{i}$ as free variables from $\mathbb{F}_{q}$. Thus, $S_{i}$ has a uniform distribution over $\mathbb{F}_{q}^{t_{i}}$, and no information about $S_{i}$ can be obtained from these $t_{i}-1$ shadows and public messages in the $i$th-stage. Consequently, these $t_{i}-1$ participants cannot restore the unrevealed secrets in polynomial time.
\end{proof}

\section{Performance analysis}
In this section, we discuss the memory consumption, time consumption, performance feature of some proposed schemes and our schemes. In order to compare these schemes, we set that there are $n$ participants, $k$ shared secrets, and the threshold is $t$. Notice that in this section we assume that $t=t_{1}=t_{2}=\cdots=t_{k}$ or $t=l_{1}=l_{2}=\cdots=l_{k}$ in our multi-stage secret sharing schemes. In all the table below, we use abbreviation OS to represent our schemes.

\subsection{Memory consumption}
At first, we compare some known schemes \cite{HD94,H95,M162,PE17} with our schemes from the perspective of memory consumption in the TABLE \uppercase\expandafter{\romannumeral 1}, where we list the number of public values in these schemes. Note that all the schemes here are multi-stage secret sharing schemes. Besides, the first three schemes can recover secrets in a predefined order, and the next few schemes can reconstruct secrets in any order.

And we take $(t,k,n)=(3,4,7),(6,7,10),(8,9,12),(11,12,14),(12,13,14)$ respectively in Fig. 1 to show the relations between the number of public values and the values of $(t,k,n)$ intuitively. As shown in Fig.1, generally speaking, the last three schemes need more memory consumption than the first three ones, which is because the last three schemes are all MSSSAOs and lattice-based secret sharing schemes. Later, we only compare the last three schemes. Among these three schemes, PE \cite{PE17} is the most efficient. Relatively speaking, OS 3 and OS 4 are less efficient than OS 1 and OS2, because they use different types of $ILR$ relations. Besides, the smaller the difference between $n$ and $t$, the more effective OS 1 and OS 2 are.

\begin{center}\scriptsize
\topcaption{\textbf{Comparison of the memory consumption in proposed schemes}
\label{Comparison of the amount of the public values in proposed schemes}}
\tablefirsthead{\hline\multicolumn{1}{|c|}{Scheme}&\multicolumn{2}{|c|}{Number of public values}
&\multicolumn{1}{|c|}{ Public values}\\\hline}
\tablelasttail{\hline}
\tablehead{\hline\multicolumn{4}{|c|}
{\small\sl continued from previous page}\\\hline
\multicolumn{1}{|c|}{Scheme}&\multicolumn{2}{|c|}{Number of public values}
&\multicolumn{1}{|c|}{Public values}\\\hline}
\tabletail{\hline\multicolumn{4}{|c|}{\small\sl continued on next page}\\\hline}

\begin{supertabular}{|p{0.1\textwidth}|p{0.1\textwidth}|p{0.15\textwidth}|p{0.59\textwidth}|}
\multirow{1}{*}{HD\cite{HD94}} & \multicolumn{2}{c|} {$kn$} &
 \makecell*[c]{$d_{ij}=z_{ij}-y_{j}=f^{i-1}(y_{j})-y_{j}$,\\
 $i=1,2,\cdots,k, j=1,2,\cdots,n$.}\\
\hline

\multirow{1}{*}{LH\cite{H95}} & \multicolumn{2}{c|} {$k(n-t)$} &
 \makecell*[c]{$h_{j}(m)$,\\
 $j=0,1,\cdots,k-1$,\\
 $m=1,2,\cdots,n-t$.}\\
\hline

\multirow{1}{*}{SM\cite{M162}} &
\multicolumn{2}{c|} {$k(n-t+2)+1$}
 & \makecell*[c]{$r,r_{ji}$\\
 $j=1,2,\cdots,k, t_{j}-1\leq i\leq n$}\\
\hline

\multirow{2}{*}{PE\cite{PE17}} & \multicolumn{2}{c|} {$2(n+k)+1$}
& \makecell*[c]{$v, A_{i}, i=1,2,\cdots,k, \lambda_{j},j=1,2,\cdots,n,$\\
  $h_{j},j=1,2,\cdots,n, H(s_{i}),i=1,2,\cdots,k.$}\\
\hline

\multirow{2}{*}{OS 1 \& 2} & \multicolumn{2}{c|} {$k(n-t+6)+(n+1)$}
& \makecell*[c]{$G_{1},G_{2},\cdots,G_{k},F, h_{1},h_{2},\cdots,h_{n},H(S_{1}),H(S_{2}),\cdots,H(S_{k})$\\
  $c_{1},c_{2},\cdots,c_{k},r_{i,j},1\leq i\leq k,\: t_{i}\leq j\leq n,$\\ $u_{i,n+1},u_{i,n+2},i=1,2,\cdots,k.$}\\
\hline

\multirow{2}{*}{OS 3\& 4} & \multicolumn{2}{c|} {$k(n+4)+(n+2)$}
& \makecell*[c]{$G_{1},G_{2},\cdots,G_{k},F, h_{1},h_{2},\cdots,h_{n},H(S_{1}),H(S_{2}),\cdots,H(S_{k})$\\
  $c,r_{i,j},1\leq i\leq k,\: l_{i}\leq j\leq n,$\\
  $u_{i,n+1},u_{i,n+2},\cdots,u_{i,n+l_{i}+1},1\leq i\leq k.$}\\
\hline
\end{supertabular}
\end{center}

\begin{figure}
  \centering
  \includegraphics[width=4in]{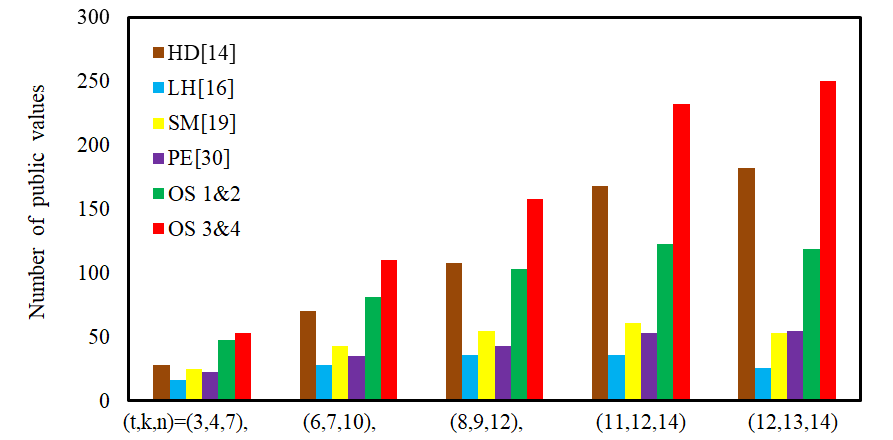}\\
  \caption{Memory consumption in proposed schemes}\label{Memory consumption of proposed schemes}
\end{figure}

\subsection{Time consumption}

In this subsection, we only compare the time consumption of our schemes with PE \cite{PE17} since they are both post-quantum multi-stage secret sharing schemes. Here, we only consider the lattice-based multi-stage secret sharing schemes. Because these two schemes are both multi-use secret sharing schemes, which means that the setup phase of these schemes needs to be performed only once, then we merely consider the time consumption of the other three phases, i.e., the construction phase, recovery phase, and verification phase. Notice that here $r_{1}\geq \max\{t\log t, n\}$ and $r_{2}\geq \max\{t\log t, \log n\}$.

\begin{center}\scriptsize
\topcaption{\textbf{Comparison of the time consumption in presented scheme}
\label{Comparison of the Computational complexity in presented scheme}}

\tablefirsthead{\hline\multicolumn{1}{|c|}{Scheme}
&\multicolumn{1}{|c|}{ Con}&\multicolumn{1}{|c|}{ Ver}&\multicolumn{1}{|c|}{ Rec}\\\hline}
\tablelasttail{\hline}
\tablehead{\hline\multicolumn{4}{|c|}
{\small\sl continued from previous page}\\\hline
\multicolumn{1}{|c|}{Scheme}
&\multicolumn{1}{|c|}{ Con}&\multicolumn{1}{|c|}{ Ver}&\multicolumn{1}{|c|}{ Rec}\\\hline}
\tabletail{\hline\multicolumn{4}{|c|}{\small\sl continued on next page}\\\hline}

\begin{supertabular}{|p{0.08\textwidth}|p{0.29\textwidth}|p{0.17\textwidth}|p{0.21\textwidth}|}

\multirow{1}{*}{PE\cite{PE17}}
 & \makecell*[c]{$O(n^{3})$} &\makecell*[c]{$O(tr_{1})$}& \makecell*[c]{$O(t^{3})$}\\

\hline
\multirow{1}{*}{OS}
& \makecell*[c]{$O(tr_{2})$}
&\makecell*[c]{$O(tr_{2})$}
&\makecell*[c]{Way1\:$O(t^{2})$\\ Way2\: $O(t)$}\\
\hline

\end{supertabular}
\end{center}

From the TABLE \uppercase\expandafter{\romannumeral 2}, since $t\leq n$ and generally $r_{1}\geq r_{2}$, we can know that our scheme is more efficient than PE \cite{PE17} in the construction phase and verification phase. What's more, in the recovery phase, the time complexity of PE \cite{PE17} is $O(t^{3})$. Nevertheless, according to different conditions, our schemes have different time cost. If any authorized set wants to restore the secrets, the time consumption of this phase in our schemes is $O(t^{2})$. Furthermore, if the indices of these participants are successive, the time complexity of this phase is reduced to only $O(t)$. From the point of view of the time consumption, our new proposed schemes are much more efficient than PE \cite{PE17}. Therefore, our schemes can be regarded as the improved versions of PE \cite{PE17}.

\subsection{Performance feature}
Finally, we analyze the performance features of the schemes in \cite{HD94,H95,DM15,M162,PE17} and our schemes in TABLE \uppercase\expandafter{\romannumeral 3}.

$\bullet$ Feature 1: Classification of MSS

$\bullet$ Feature 2: The secret recovery phase can be performed independently

$\bullet$ Feature 3: The recovered secrets will reveal the information of the remaining secrets

$\bullet$ Feature 4: Multiple secrets can be restored

$\bullet$ Feature 5: Every participant merely keeps one share

$\bullet$ Feature 6: The scheme is a multi-use secret sharing scheme

$\bullet$ Feature 7: The size of every share or shadow is the same as every secret

$\bullet$ Feature 8: Only one secret can be reconstructed in one stage

$\bullet$ Feature 9: The shares provided by the dealer can be verified

$\bullet$ Feature 10: The correctness of recovered secrets can be verified

$\bullet$ Feature 11: The secret sharing scheme is whether post-quantum or not

$\bullet$ Feature 12: Have different ways for recovery

In order to save space, as for the feature 1, we use PO to represent MSSSPO, and AO to represent MSSSAO. Besides, Y denotes YES, and N denotes NO.

\begin{table}[!htbp]\small
\caption{Performance\, feature of proposed schemes}
\label{tab:1}
\scalebox{1.25}[1.05]{
\begin{tabular}{|c|c|c|c|c|c|c|}
\hline
Feature  & HD\cite{HD94} & LH\cite{H95} & SM\cite{M162} &  MD\cite{DM15} & PE\cite{PE17} & OS \\
\hline
1 & PO & PO & PO & SMSS & AO & AO \\
\hline
2 & N & N & N & - & Y & Y \\
\hline
3  & Y & Y & Y & - & N & N \\
\hline
4 & Y & Y & Y & Y & Y & Y \\
\hline
5 & Y & Y & Y & Y& Y & Y\\
\hline
6 & N & N & Y & Y & Y & Y \\
\hline
7 & Y & Y & Y & Y & Y & Y \\
\hline
8 & Y & Y & Y & - & Y & Y \\
\hline
9 & N & N & N & Y& Y & Y \\
\hline
10 & N & N & N & N & Y & Y \\
\hline
11 & N & N & N & N & Y & Y \\
\hline
12 & N & N & Y & Y & N & Y \\
\hline
\end{tabular}}
\end{table}

From TABLE \uppercase\expandafter{\romannumeral 3}, according to different categories of these schemes, we know that the first three schemes can restore secrets in a predefined order, the last two schemes can restore secrets in any order, and the MD \cite{DM15} schemes can restore all secrets at the same time. Therefore, only the last two schemes can restore secrets independently, for the reason that in the first three schemes, the recovered secret will reveal some information of the remaining secrets, however the last two schemes do not have this drawback. As for the MD scheme \cite{DM15}, it restores all secrets simultaneously, so we do not consider the order of its secret recovery.

Notice that all these schemes mentioned here are multi-secret sharing schemes. However, these participants only need keep one share to recover all secrets. Furthermore, if their shares are masked in the recovery phase, the shares can be utilized more than one time and these schemes become multi-use secret sharing schemes. So, in the first two schemes, the shares are not reusable. Whereas, in our schemes, if we want to share a new set of secrets, what we need to do is only to release new matrices $G_{i}\in \mathbb{F}_{q}^{t_{i} \times r}$ for $i=1,2,\cdots,k$, which does not need the secure channel between the dealer and the participants. Then we can obtain new shadows using the original shares and restore the new secrets by the ways mentioned in the recovery phase of our schemes. What's more, the feature 7 means that these schemes are all ideal threshold secret sharing schemes. In addition, all the schemes can recover only one secret in one stage except the scheme \cite{DM15}, which is because this scheme reconstruct all secrets in only one stage.

As for the features 9 and 10, only PE \cite{PE17} and our schemes can not only detect the cheat from the dealer, but also detect the dishonest participants. More importantly, the last two schemes are lattice-based secret sharing schemes, which means that they can resist the attack by quantum computation.

When it comes to the feature 12, it seems that our schemes have the same advantages as \cite{M162} and \cite{DM15}. Nevertheless, we provide a general form of the $ILR$ relations used in \cite{M162} and \cite{DM15}. In fact, we can use an $ILR$ relation with degree $t-1$ to substitute the $ILR$ relation used in our presented schemes. Then, we shouldn't publish corresponding $c_{i}$ or $c$ to make our schemes in every stage still a $(t,n)$ threshold secret sharing schemes. Because $c_{i}$ or $c$ can be computed from $t$ terms of the $ILR$ relation with degree $t$. However, in order to keep the consistence between the threshold of scheme in every stage and corresponding degree of $ILR$ relation, we choose to use a $t$-degree $ILR$ relation. Meanwhile, we publish related $c_{i}$ or $c$ so that any at least $t$ values can reconstruct the shared secrets by running the algorithm proposed in our schemes. Similarly, we can also use an $ILR$ relation with degree $t+1$ in our schemes, but we need release more public values.

On the other hand, although the SM \cite{M162}, MD \cite{DM15} and our schemes belong to different categories, our schemes can be changed into the other two types of MSSs easily through simple modifications. For example, if we want to turn our MSSSAOs into MSSSPOs, we can use the model proposed in Section 2.2 of \cite{M162}. Similarly, if we want to turn our MSSSAOs into SMSSs, we can use the way proposed in \cite{YF201}. Accordingly, using our method, we can construct a series of different types of MSSs based on different security requirements and application scenarios, which means that the method proposed in our schemes can be also applied to MSSSPOs and SMSSs.

Therefore, our proposed schemes are better than the schemes mentioned above.

\section{Conclusion}

In this paper, we utilize the inhomogeneous linear recursion and Ajtai's function to construct four verifiable post-quantum multi-stage secret sharing schemes with recovering secrets in any order.

At first, we conclude a general formula of some $ILR$ relations used in the literature, and divide them into two types, i.e., Type-$t$ and Type-$l$ $ILR$ relations. Then, we use special cases of $ILR$ relations and Ajtai's function to construct four different MSSSAOs. Later, we give the security analysis from three aspects, i.e., correctness, verifiability and privacy. Importantly, we prove that the recovered secrets cannot leak any information of the unrecovered secrets in our four schemes with computational security. Besides, we show that any number of participants less than the threshold cannot restore unrevealed secret.

Compared with some known schemes, our schemes can verify the authenticity of both shares and recovered secrets. Meanwhile, the shares are reusable in our schemes. Moreover, we have several ways to recover the secrets due to the properties of $ILR$ relations. Since the universality of our method, our schemes can be easily changed into SMSSs and MSSSPOs, which means that we can use our method to construct a series of different types of novel MSSs readily. Although our schemes need more memory consumption than some known schemes, we need much less time consumption. Notice that the first two schemes we proposed are better than the last two schemes from the perspective of the memory consumption. Furthermore, because we utilize Ajtai's function, a collision resistant lattice-based function, our four schemes can resist the attack from quantum computation, which can be applied in more scenarios in the future.



\end{document}